\documentclass[journal]{IEEEtran}

\usepackage{amsmath}
\usepackage{amssymb}
\usepackage{upref}
\usepackage{amsfonts}
\usepackage{graphicx}
\usepackage{epic}
\usepackage{psfrag}
\usepackage{verbatim}

\usepackage{amsfonts}
\usepackage{amssymb}

\usepackage[usenames]{color}
\definecolor{dark_purple}{cmyk}{0.0, 1.0, 0.0, 0.6}
\definecolor{OXO-emph}{RGB}{200,127,0}
\definecolor{dark_green}{RGB}{10,127,50}

\usepackage{slashbox}

\newcommand{\deff}{\mbox{$\stackrel{\rm def}{=}$}}

\newcommand{\floorenv}[1]{\left\lfloor #1 \right\rfloor}
\newcommand{\sbinom}[2]{\left[ \begin{array}{c} #1 \\ #2 \end{array} \right] }
\newcommand{\sbinomq}[2]{\sbinom{#1}{#2}_q }

\newcommand{\field}[1]{\mathbb{#1}}
\newcommand{\C}{\field{C}}
\newcommand{\F}{\field{F}}

\newcommand{\V}{\field{V}}

\newcommand{\dS}{\field{S}}
\newcommand{\T}{\field{T}}
\newcommand{\R}{\field{R}}

\newcommand{\cF}{{\cal F}}
\newcommand{\cH}{{\cal H}}
\newcommand{\cA}{{\cal A}}
\newcommand{\cB}{{\cal B}}
\newcommand{\cC}{{\cal C}}
\newcommand{\cG}{{\cal G}}
\newcommand{\cL}{{\cal L}}

\newcommand{\cS}{{\cal S}}

\newcommand{\cP}{{\cal P}}

\newcommand{\sP}{\cP}
\newcommand{\sG}{\cG}

\newcommand{\Gr}{\smash{{\sG\kern-1.5pt}_q\kern-0.5pt(n,k)}}
\newcommand{\Grr}{\smash{{\sG\kern-1.5pt}_q\kern-0.5pt(n,r)}}
\newcommand{\Gfourk}{\smash{{\sG\kern-1.5pt}_q\kern-0.5pt(4k,2k)}}
\newcommand{\Gk}{\smash{{\sG\kern-1.5pt}_q\kern-0.5pt(n,k_1)}}
\newcommand{\Gkk}{\smash{{\sG\kern-1.5pt}_q\kern-0.5pt(n,k_2)}}
\newcommand{\Grtwo}{\smash{{\sG\kern-1.5pt}_2\kern-0.5pt(n,k)}}
\newcommand{\Gkone}{\smash{{\sG\kern-1.5pt}_q\kern-0.5pt(n,k_1)}}
\newcommand{\Gktwo}{\smash{{\sG\kern-1.5pt}_q\kern-0.5pt(n,k_2)}}
\newcommand{\Ps}{\smash{{\sP\kern-2.0pt}_q\kern-0.5pt(n)}}
\newcommand{\Span}[1]{{\left\langle {#1} \right\rangle}}

\newcommand{\CMRD}{\C^{\textmd{MRD}}}

\newtheorem{theorem}{Theorem}
\newtheorem{corollary}{Corollary}
\newtheorem{lemma}[theorem]{Lemma}

\newtheorem{remark}{Remark}
\newtheorem{example}{Example}

\newcommand{\Gauss}[2]{\begin{footnotesize}\left[\begin{array}
{c}#1\\#2\end{array}\right]_{q}\end{footnotesize}}

\newcommand{\GaussBin}[2]{\begin{footnotesize}\left[\begin{array}
{c}#1\\#2\end{array}\right]_{2}\end{footnotesize}}

\DeclareMathOperator{\rank}{rank}
\begin{document}

\bibliographystyle{IEEEtran}

\title{Codes and Designs Related to Lifted MRD Codes}
\author{Tuvi
Etzion,~\IEEEmembership{Fellow,~IEEE} and Natalia Silberstein
\thanks{T. Etzion is with the Department of Computer Science,
Technion --- Israel Institute of Technology, Haifa 32000, Israel.
(email: etzion@cs.technion.ac.il).}
\thanks{N. Silberstein is with the Department of Electrical
and Computer Engineering, University of Texas at Austin, Austin, TX 78712-1684,
USA (email:natalys@austin.utexas.edu).  This work was done when she
was with the Department of Computer Science,
Technion --- Israel Institute of Technology, Haifa 32000, Israel.
This work is part of her Ph.D.
thesis performed at the Technion.}
\thanks{The material in this paper was presented in part in the 2011 IEEE
International Symposium on Information Theory, Saint Petersburg,
Russia, August 2011.}
\thanks{This work was supported in part by the Israel
Science Foundation (ISF), Jerusalem, Israel, under Grant 230/08.}
}

\maketitle

\begin{abstract}
Lifted maximum rank distance (MRD) codes, which are constant
dimension codes, are considered. It is shown that a lifted MRD
code can be represented in  such a way that it forms a block
design known as a transversal design. A slightly different
representation of this design makes it similar to a $q-$analog of
a transversal design. The structure of these designs is used to obtain upper
bounds on the sizes of constant dimension codes which contain a
lifted MRD code. Codes which attain these bounds are constructed.
These codes are the largest known codes for the given parameters.
These transversal designs can be also used to derive a new
family of linear codes in the Hamming space. Bounds on the minimum distance
and the dimension of such codes are given.
\end{abstract}

\begin{keywords}
constant dimension codes, Grassmannian space, lifted
MRD codes, rank-metric codes, transversal designs.
\end{keywords}

\section{Introduction}
\label{sec:introduction}

\noindent\looseness=-1 \PARstart{L}{et} $\F_q$ be the finite field of size $q$.
For two $k \times \ell$ matrices $A$ and $B$ over $\F_q$ the {\it
rank distance} is defined by
$$
d_R (A,B) \deff \text{rank}(A-B)~.
$$
A  $[k \times \ell,\varrho,\delta]$ {\it rank-metric code} $\cC$
is a linear code, whose codewords are $k \times \ell$ matrices
over $\F_q$; they form a linear subspace with dimension $\varrho$
of $\F_q^{k \times \ell}$, and for each two distinct codewords $A$
and $B$ we have that $d_R (A,B) \geq \delta$. For a
${[k \times \ell,\varrho,\delta]}$ rank-metric code $\cC$ it was proved
in~\cite{Del78,Gab85,Rot91} that
\begin{equation}
\label{eq:MRD} \varrho \leq
\text{min}\{k(\ell-\delta+1),\ell(k-\delta+1)\}~.
\end{equation}
This bound, called Singleton bound for the rank metric,
is attained for all
feasible parameters. The codes which attain this bound are called {\it
maximum rank distance} codes (or MRD codes in short).


Rank-metric codes have found application in public key cryptosystems~\cite{GPT91}, space-time coding~\cite{LGB03},
authentication codes~\cite{WXS-N03},
rank-minimization over finite fields~\cite{TBD11}, and distributed storage systems~\cite{SRV12}.
Recently, rank-metric
codes also have found a new application in the construction of
error-correcting codes for random network coding~\cite{SKK08}.
For this application, the $k \times \ell$ matrices are
lifted into $k$-dimensional subspaces of
$\F_q^{k+\ell}$~\cite{SKK08} as described below.

Let $A$ be a $k \times \ell$ matrix over $\F_q$ and let $I_k$ be a
$k \times k$ identity matrix. The matrix $[ I_k ~ A ]$ can be
viewed as a generator matrix of a $k$-dimensional subspace of
$\F_q^{k+\ell}$, and it is called the \emph{lifting} of $A$~\cite{SKK08}.
\begin{example}
Let $A$ and $[I_3 ~A]$ be the following matrices over $\F_2$
$$A=\left( \begin{array}{ccc}
1& 1 & 0\\
0& 1 & 1\\
0& 0 & 1
\end{array}
\right) , ~[I_3 ~ A] =\left( \begin{array}{cccccc}
1&0&0&1& 1 & 0\\
0&1&0&0& 1 & 1\\
0&0&1&0& 0 & 1
\end{array}
\right),
$$
then the subspace obtained by the lifting of $A$ is given by the
following $8$ vectors:
$$(1 0 0 1 1 0),
(0 1 0 0 1 1), (0 0 1 0 0 1),(1 1 0 1 0 1),$$
$$(1 0 1 1 1 1),(0 1 1 0 1 0),
(1 1 1 1 0 0),(0 0 0 0 0 0).
$$
\end{example}

Given a nonnegative integer $k \leq n$, the set of all
$k-$dimensional subspaces of \smash{$\F_q^n$} forms the
\emph{Grassmannian space} (Grassmannian in short) over $\F_q$,
which is denoted by $\Gr$. It is well known that
$|\Gr|=\Gauss{n}{k}=\prod_{i=0}^{k-1}\frac{q^{n-i}-1}{q^{k-i}-1}$,
where $\Gauss{n}{k}$ is the $q$-\emph{ary Gaussian coefficient}.

A subset $\C$ of $\Gr$ is called an  $(n,M,d_S,k)_q$ {\it constant
dimension code} if it has size $M$ and minimum subspace distance $d_S$,
where the distance function in $\Gr$ is defined by
\begin{equation*}
\label{def_subspace_distance} d_S (X,\!Y) \deff  \dim X + \dim Y
-2 \dim\bigl( X\, {\cap}Y\bigr),
\end{equation*}
for any two subspaces $X$ and $Y$ in $\Gr$.
$\cA_q(n,d,k)$ will denote the maximum size of an $(n,M,d,k)_q$ code.

Codes in the Grassmannian gained recently lot of interest due to
the work by Koetter and Kschischang~\cite{KK}, where they
presented an application of such codes for error-correction in
random network coding.
When the codewords of a rank-metric code
$\cC$ are lifted to $k$-dimensional subspaces, the result is
a constant
dimension code $\C$. If $\cC$ is an MRD code
then $\C$ is called a \emph{lifted MRD
code}~\cite{SKK08}. This code will be denoted by $\CMRD$.

\begin{theorem}\cite{SKK08}
\label{trm:param lifted MRD}
Let $k$, $n$ be positive integers such that ${k \leq n-k}$.
If $\cC$ is a $[k \times (n-k), (n-k)(k-\delta +1),\delta ]$ MRD
code then $\CMRD$ is an $(n,q^{(n-k)(k-\delta+1)},
2\delta, k)_{q}$ code.
\end{theorem}

In view of Theorem~\ref{trm:param lifted MRD},
we will assume throughout the paper that that $1 < k \leq n-k$.
$\CMRD$ which is an $(n,q^{(n-k)(k-\delta+1)},
2\delta, k)_{q}$ code will be also called an $(n,k,\delta)_q$ $\CMRD$.
If no parameters for $\CMRD$ will be given we will assume it is
an $(n,k,\delta)_q$ $\CMRD$.

Most of the constructions for large constant dimension codes known in
the literature produce codes which contain
$\CMRD$~\cite{EtSi09,GaYa10,MGR08,SiEt10,SKK08,Ska10,TrRo10}. The
only constructions which generate codes that do not contain
$\CMRD$ are given in~\cite{EV08,KoKu,TMR10}. These constructions are
either of so called orbit codes or specific constructions for
small parameters. Moreover, only $(n,M,d ,3)_2$ orbit codes
(specifically cyclic codes) with $8 \leq n \leq 12$, and
$(6,77,4,3)_2$ and $(7,304,4,3)_2$ codes are the largest codes for
their specific parameters which do not contain
$\CMRD$~\cite{KoKu}. This motivates the question, what is the
largest constant dimension code which contain $\CMRD$?

The well-known concept of $q$-analogs replaces subsets by
subspaces of a vector space over a finite field and their orders
by the dimensions of the subspaces. In particular, the $q
\text{-analog}$ of a constant weight code in the Johnson space is
a constant dimension code in the Grassmannian space. Related to
constant dimension codes are $q$-analogs of block designs.
$q$-analogs of designs were studied
in~\cite{AAK01,Bra05,EV08,EV10,ScEt02,Tho87}. For example,
in~\cite{AAK01} it was shown that Steiner structures (the
$q$-analog of Steiner system), if exist, yield optimal codes in
the Grassmannian. Another connection is the constructions of
constant dimension codes from spreads which are given
in~\cite{EV08} and~\cite{MGR08}.

In this paper we consider several topics related to
lifted MRD codes. First, we discuss properties of these codes related to block
designs. We prove that the codewords of $\CMRD$ form a design
called a transversal design, a structure which is known to be
equivalent to the well known orthogonal array. We also prove that
the same codewords form a subspace transversal design, which is
akin to the transversal design, but not its $q$-analog.

The structure of $\CMRD$ as a transversal design leads to the other
results given in this paper. We derive for new lower bounds
on $\cA_q (n,d,k)$ and upper bounds on the sizes of
error-correcting constant dimension codes which contain $\CMRD$.
In particular, we prove
that if an $(n,M,2(k-1),k)_q$ code $\mathbb{C}$, $k \geq 3$,
contains an $(n,k,k-1)_q$ $\CMRD$ code then

\begin{small}
$$M\leq q^{2(n-k)}+ \cA_q (n-k,2(k-2),k-1)~.$$ \end{small}
\noindent We present a construction for codes which either
attain this bound or almost attain it for $k=3$. These codes are
the largest known $(n,M,4,3)_q$ codes for $n \geq 13$.

We prove that if an $(n,M,2k,2k)_q$ code $\mathbb{C}$ contains an
$(n,2k,k)_q$ $\CMRD$ code then

\begin{small} $$M\leq q^{(n-2k)(k+1)}+ \Gauss{n-2k}{k}
\frac{q^n-q^{n-2k}}{q^{2k}-q^k} + \cA_q (n-2k, 2k,2k)~.$$
\end{small}
\noindent We present a construction for codes which attain this
bound when $2k=4$, $n=8$, and for all $q$. These codes are the
largest known for the related parameters.

The incidence matrix of the transversal design derived from
$\CMRD$ can be viewed as a parity-check matrix of a linear code in
the Hamming space. This way to construct a linear code from a
design is
well-known~\cite{AHKXL04,Dou93,JoWe01,KV03,KLF01,LaMi07,LTLMH08,VKK02,VaMi04,ZHLA-G10}.
We find the properties of these codes, in particular, we present the bounds on their minimum distance and dimension.

The rest of this paper is organized as follows. In
Section~\ref{sec:MRD} we present properties of lifted
MRD codes. Then we prove that
these codes form transversal designs in sets and subspaces.  In
Section~\ref{sec:upbound} we discuss some known upper bounds on
$\cA_q(n,d,k)$ and present two new upper bounds on the sizes of
constant dimension codes which contain $\CMRD$. In
Sections~\ref{sec:construct} and~\ref{sec:constructionB} we
provide constructions of two families of codes that attain the
upper bounds of Section~\ref{sec:upbound}.
In Section~\ref{sec:Linear codes} we consider properties of linear
codes whose parity-check matrices are derived from $\CMRD$.
Conclusions and
problems for future research are given in
Section~\ref{sec:conclude}.

\vspace{-0.1cm}

\section{Lifted MRD codes and transversal designs}
\label{sec:MRD}
In this section we prove that a lifted MRD code yield
a combinatorial structure known as a transversal design.
Moreover, the codewords of these codes form the blocks of a new
type of transversal design, called a subspace transversal design.
Based on these designs, we will present some novel results
in the following sections. We first examine some combinatorial properties
of lifted MRD codes. Based on these properties we
will construct the transversal designs.

\subsection{Properties of lifted MRD codes}

Let $\cL^{(n,k)}$ be the set of $q^{n}-q^{n-k}$ vectors of length $n$ over
$\F_q$ in which not all the first $k$ entries are {\it zeroes}.
The following lemma is a simple observation.

\begin{lemma}
\label{lem:CMRD=L} All the nonzero vectors which are contained in
codewords of an $(n,k,\delta)_q$ $\CMRD$ belong to $\cL^{(n,k)}$.
\end{lemma}

For a set $\cS \subseteq \F_q^n$, let $\Span{\cS}$ denotes the
subspace of $\F_q^n$ spanned by the elements of $\cS$. If
$\cS=\{v\}$ is of size one then we denote $\Span{\cS}$ by
$\Span{v}$.
For $v_1\in\F_q^{n_1}$ and $v_2\in\F_q^{n_2}$ we denote by $v_1||v_2\in \F_q^{n_1+n_2}$
the concatenation of $v_1$ and~$v_2$.
 Let $\mathbb{V}^n=\{\Span{v}:v\in \cL^{(n,k)}\}$ be the set of all
$\frac{q^n-q^{n-k}}{q-1}$ one-dimensional subspaces of $\F_q^n$
whose nonzero vectors are contained in $\cL^{(n,k)}$. We identify each
subspace  ${A \in \mathcal{G}_q(\ell,1)}$, for any given
$\ell$, with the vector ${v_A \in A}$ (of length~$\ell$) in
which the first nonzero entry is a {\it one}.

For each $A\in \mathcal{G}_q(k,1)$ we define
$$
\mathbb{V}_A^{(n,k)} \deff \{X \;|\;X= \Span {  v },\;v\;=\;v_A||z ,\; z \in
\F_q^{n-k} \},
$$
in other words, $\mathbb{V}_A^{(n,k)}$ consists of all one-dimensional subspaces whose restriction to
the first $k$ coordinates is precisely~$A$.
$\{ \mathbb{V}_A^{(n,k)} : A \in \mathcal{G}_q(k,1)\}$ contains
$\frac{q^k-1}{q-1}$ sets, each one of the size $q^{n-k}$. These
sets partition the set $\mathbb{V}^n$, i.e.,

$$
\V_A^{(n,k)} \cap \V_B^{(n,k)} = \varnothing ,~ A,B \in \mathcal{G}_q(k,1) ,~ A \neq B,
$$
and
$$
\mathbb{V}^n=\bigcup_{A\in\mathcal{G}_q(k,1)}\mathbb{V}_A^{(n,k)}~.
$$
We say that a vector $v\in \F_q^n$ is in $\mathbb{V}_A^{(n,k)}$  if $v\in X$
for ${X\in \mathbb{V}_A^{(n,k)}}$. Clearly, $\Span { \{ v_A||z',\;v_A||z'' \}
}$, for $A \in \mathcal{G}_q(k,1)$ and $z' \neq z''$, contains a
vector with $k$ leading \emph{zeroes}. Such a vector does not
belong to $\cL^{(n,k)}$ and hence, by Lemma~\ref{lem:CMRD=L} we have

\begin{lemma}
\label{lem:CMRD=VA} For each $A\in \mathcal{G}_q(k,1)$, a codeword
of $\CMRD$ contains at most one element from
$\mathbb{V}_A^{(n,k)}$.
\end{lemma}

Note that each $k$-dimensional subspace of $\F_q^n$ contains
$\Gauss{k}{1}=\frac{q^k-1}{q-1}$ one-dimensional subspaces.
Therefore, by Lemma~\ref{lem:CMRD=L}, each codeword of
$\CMRD$ contains $\frac{q^k-1}{q-1}$ elements of
$\mathbb{V}^n$. Hence, by Lemma~\ref{lem:CMRD=VA} and since
$|\mathcal{G}_q(k,1)|=\frac{q^k-1}{q-1}$ we have

\begin{corollary}
\label{cor:exactlyone} For each $A\in \mathcal{G}_q(k,1)$, a
codeword of $\CMRD$ contains exactly one element from
$\mathbb{V}_A^{(n,k)}$.
\end{corollary}

\begin{lemma}\label{lm:k-delta+1}
Each  $(k-\delta+1)$-dimensional subspace $Y$ of $\F_q^n$, whose
nonzero vectors are contained in $\cL^{(n,k)}$, is contained in exactly
one codeword of an $(n,k,\delta)_q$ $\CMRD$.
\end{lemma}
\begin{proof}
Let
$$
\dS \deff \{Y\in \mathcal{G}_q(n,k-\delta+1):~ |Y \cap \cL^{(n,k)} | =
q^{k-\delta+1}-1 \},
$$
i.e. $\dS$ consists of all $(k-\delta+1)$-dimensional subspaces
of $\F_q^n$ in which all the nonzero vectors
are contained in $\cL^{(n,k)}$.

Since the minimum distance of $\CMRD$ is $2\delta$ and
its codewords are $k$-dimensional subspaces, it follows that the
intersection of any two codewords is at most of dimension
$k-\delta$. Hence, each $(k-\delta+1)$-dimensional subspace of
$\F_q^n$ is contained in at most one codeword. The size of
$\CMRD$ is $q^{(n-k)(k-\delta+1)}$, and the number of
$(k-\delta+1)$-dimensional subspaces in a codeword is exactly
$\Gauss{k}{k-\delta+1}$. By Lemma~\ref{lem:CMRD=L},
each $(k-\delta+1)$-dimensional subspace,
of a codeword, is contained in $\dS$. Hence, the
codewords of $\CMRD$ contain exactly
$\Gauss{k}{k-\delta+1}q^{(n-k)(k-\delta+1)}$ distinct
$(k-\delta +1)$-dimensional subspaces of $\dS$.

To complete the proof we only have to show that $\dS$ does not
contain more $(k-\delta+1)$-dimensional subspaces. Hence, we will
compute the size of $\dS$. Each element of $\dS$ intersects with
each $\mathbb{V}_{A}^{(n,k)}$, $A\in\mathcal{G}_q(k,1)$ in at most one
one-dimensional subspace (since it contains vectors
only from $\cL^{(n,k)}$). There are $\Gauss{k}{k-\delta+1}$ ways to
choose an arbitrary ${(k-\delta+1)}$-dimensional subspace of
$\F_q^k$. For each such subspace we choose a basis
$\{ x_1 , x_2 , \ldots , x_{k - \delta +1} \}$, where each
$x_i$ belongs to a different set $\V_A^{(n,k)}$, $A\in\mathcal{G}_q(k,1)$
(clearly, by previous definition, in each such basis vector the first
nonzero entry is a {\it one}).
A basis for a
$(k-\delta+1)$-dimensional subspace of $\dS$ will be generated by
concatenation of $x_i$ with a vector $z\in\F_q^{n-k}$ for each
$i$, $1 \leq i \leq k-\delta+1$.
Therefore, there are $q^{(n-k)(k-\delta+1)}$ ways to choose a basis for
an element of~$\dS$. Hence,
$|\dS|=\Gauss{k}{k-\delta+1}q^{(n-k)(k-\delta+1)}$.

Thus, the lemma follows.
\end{proof}

\begin{corollary}
\label{cor:k-delta-i} For each $i$, $0 \leq i \leq k-\delta -1$,
each ${(k-\delta-i)}$-dimensional subspace of
$\F_q^n$, whose nonzero vectors are contained in $\cL^{(n,k)}$, is
contained in exactly $q^{(n-k)(i+1)}$ codewords of
$\CMRD$.
\end{corollary}
\begin{proof}
The size of $\CMRD$ is $q^{(n-k)(k-\delta +1)}$.
The number of $(k-\delta-i)$-dimensional subspaces in a codeword
is exactly $\Gauss{k}{k-\delta-i}$. Hence, the total number
of $(k-\delta -i)$-dimensional subspaces in $\CMRD$ is
$\Gauss{k}{k-\delta -i} q^{(n-k)(k-\delta+1)}$ (clearly, each such
$(k-\delta-i)$-dimensional
subspace is counted more than once in this computation). Similarly to the
proof of Lemma~\ref{lm:k-delta+1}, we can prove that the total
number of $(k-\delta-i)$-dimensional subspaces which contain
nonzero vectors only from $\cL^{(n,k)}$ is $\Gauss{k}{k-\delta -i}
q^{(n-k)(k-\delta-i)}$. By simple symmetry, each two such subspaces
are contained in the same number of codewords of $\CMRD$.
Thus, each $(k-\delta-i)$-dimensional
subspace of $\F_q^n$, whose nonzero vectors are contained in
$\cL^{(n,k)}$, is contained in exactly
$$
\frac{\Gauss{k}{k-\delta -i} q^{(n-k)(k-\delta
+1)}}{\Gauss{k}{k-\delta -i} q^{(n-k)(k-\delta-i)}}=
q^{(n-k)(i+1)}
$$
codewords of $\CMRD$.
\end{proof}

\begin{corollary}
\label{cor:2} Each one-dimensional subspace $X\in \mathbb{V}^n$ is
contained in exactly $q^{(n-k)(k-\delta)}$ codewords of
$\CMRD$.
\end{corollary}

By applying Corollary~\ref{cor:k-delta-i} with $k-\delta-i=2$
we also infer the following result.
\begin{corollary}
\label{cor:two_elements} Any two elements $X_1,X_2\in \mathbb{V}^n$, such that
$X_1\in \mathbb{V}_A^{(n,k)}$ and  $X_2\in \mathbb{V}_B^{(n,k)}$, $A \neq B$, are
contained in exactly $q^{(n-k)(k-\delta-1)}$ codewords of
$\CMRD$.
\end{corollary}

For the following lemma we need a generalization
of the definition of a rank-metric code to
a \emph{nonlinear} rank-metric code, which is a subset of $\F_q^{k \times
\ell}$ with minimum distance $\delta$ and size $q^\varrho$. If
$\varrho = \text{min}\{k(\ell-\delta+1),\ell(k-\delta+1)\}$, then
such a code will be also called an MRD code.

\begin{lemma}\label{lm:resolv}
$\CMRD$ can be partitioned into $q^{(n-k)(k-\delta)}$
sets, called parallel classes, each one of size $q^{n-k}$, such
that in each parallel class each element of $\mathbb{V}^n$ is
contained in exactly one codeword.
\end{lemma}

\begin{proof}
First we prove that a lifted MRD code contains a lifted MRD
subcode with disjoint codewords (subspaces). Let $G$ be the
generator matrix of a $[k \times (n-k),(n-k)(k-\delta +1),
\delta]$ MRD code $\cC$~\cite{Gab85}, $n-k \geq k$. Then $G$ has the
following form
\[G=\left(
      \begin{array}{cccc}
        g_1 & g_2 & \ldots & g_k \\
        g_1^{q} & g_2^{q} & \ldots & g_k^{q} \\
        \vdots & \vdots & \cdots & \vdots \\
        g_1^{q^{k-\delta}} &  g_2^{q^{k-\delta}} & \ldots  &  g_k^{q^{k-\delta}} \\
      \end{array}
    \right) ~,
\]
where $g_i\in \mathbb{F}_{q^{n-k}}$ are linearly independent over
$\mathbb{F}_q$.
If the last $k - \delta$ rows  are removed
from $G$, the result is an MRD subcode of $\cC$  with the minimum
distance $k$. In other words, an
$[k \times (n-k),n-k, k]$ MRD  subcode $\tilde{\cC}$
of $\cC$ is obtained. The corresponding lifted code is an $(n,q^{n-k},2k,k)_q$
lifted MRD subcode of $\C^{\text{MRD}}$.

Let $\tilde{\cC}_1 = \tilde{\cC},~\tilde{\cC}_2 , \ldots ,
~\tilde{\cC}_{q^{(n-k)(k - \delta )}}$ be the $q^{(n-k)(k - \delta
)}$ cosets of $\tilde{\cC}$ in $\cC$. All these $q^{(n-k)(k -
\delta )}$ cosets are nonlinear rank-metric codes with the same
parameters as the ${[k \times (n-k),n-k, k]}$ MRD code. Therefore,
their lifted codes form a partition of an $(n,k,\delta)_q$ $\CMRD$ into
$q^{(n-k)(k-\delta)}$ parallel classes each one of size $q^{n-k}$,
such that each element of $\mathbb{V}^n$ is contained in exactly one
codeword of each parallel class.
\end{proof}

\vspace{-0.2cm}

\subsection{Transversal designs from lifted MRD codes}

A \emph{transversal design} of groupsize $m$, blocksize $k$,
\emph{strength}~$t$ and \emph{index} $\lambda$,  denoted by
$\text{TD}_{\lambda}(t, k, m)$ is a triple
$(V,\mathcal{G},\mathcal{B})$, where

\begin{enumerate}
\item $V$ is a set of $km$ elements (called \emph{points});

\item $\mathcal{G}$ is a partition of $V$ into $k$ classes
(called \emph{groups}), each one of size $m$;

\item $\mathcal{B}$ is a collection of $k$-subsets of $V$
(called \emph{blocks});

\item each block meets each group in exactly one point;

\item each $t$-subset of points that meets each group in at most
one point is contained in exactly $\lambda$ blocks.
\end{enumerate}

When $t=2$, the strength is usually not mentioned, and the design
is denoted by $\text{TD}_{\lambda}(k,m)$.
A $\text{TD}_{\lambda}(t,k,m)$ is \emph{resolvable} if the set $\mathcal{B}$
can be partitioned into sets $\mathcal{B}_1,...,\mathcal{B}_s$,
where each element of $V$ is contained in exactly one block of
each $\mathcal{B}_i$. The sets $\mathcal{B}_1,...,\mathcal{B}_s$ are called
\emph{parallel classes}.

\begin{example}
\label{ex:TD(3,4)}
Let $V=\{1,2,\ldots,12\}$; $\mathcal{G}=\{G_1,G_2,G_3\}$, where
$G_1=\{1,2,3,4\}$, $G_2=\{5,6,7,8\}$, and $G_3=\{9,10,11,12\}$;
$\mathcal{B}=\{B_1,B_2,\ldots, B_{16}\}$, where
$B_1=\{1,5,9\}$, $B_2=\{2,8,11\}$, $B_3=\{3,6,12\}$, $B_4=\{4,7,10\}$,
$B_5=\{1,6,10\}$, $B_6=\{2,7,12\}$, $B_7=\{3,5,11\}$, $B_8=\{4,8,9\}$,
$B_9=\{1,7,11\}$, $B_{10}=\{2,6,9\}$, $B_{11}=\{3,8,10\}$, $B_{12}=\{4,5,12\}$,
$B_{13}=\{1,8,12\}$, $B_{14}=\{2,5,10\}$, $B_{15}=\{3,7,9\}$, and $B_{16}=\{4,6,11\}$.
These blocks form a resolvable $TD_1(3,4)$ with  four parallel classes
$\mathcal{B}_1=\{B_1,B_2,B_3,B_4\}$, $\mathcal{B}_2=\{B_5,B_6,B_7,B_8\}$,
$\mathcal{B}_3=\{B_9,B_{10},B_{11},B_{12}\}$, and $\mathcal{B}_4=\{B_{13},B_{14},B_{15},B_{16}\}$.
\end{example}

\begin{theorem}\label{trm:CDC=TD}
The codewords of an $(n,k,\delta)_q$
$\CMRD$ form the blocks of a resolvable transversal
design $\text{TD}_{\lambda}(\frac{q^k-1}{q-1},\;q^{n-k})$,
$\lambda=q^{(n-k)(k-\delta-1)}$, with $q^{(n-k)(k-\delta)}$
parallel classes, each one of size $q^{n-k}$.
\end{theorem}
\begin{proof}
Let $\mathbb{V}^n$ be the set of $\frac{q^n - q^{n-k}}{q-1}$ points
for the design. Each set $\mathbb{V}_A^{(n,k)}$, $A \in
\mathcal{G}_q(k,1)$, is defined to be a group, i.e., there are
$\frac{q^k-1}{q-1}$ groups, each one of size $q^{n-k}$. The
$k$-dimensional subspaces (codewords) of $\CMRD$ are the
blocks of the design. By Corollary~\ref{cor:exactlyone}, each
block meets each group in exactly one point. By
Corollary~\ref{cor:two_elements}, each 2-subset which meets each group in at
most one point is contained in exactly $q^{(n-k)(k-\delta-1)}$
blocks. Finally, by Lemma~\ref{lm:resolv} the design is resolvable
with $q^{(n-k)(k-\delta)}$ parallel classes, each one of size
$q^{n-k}$.
\end{proof}

An $N\times k$ array $\cA$ with entries from a set of $s$ elements
is an \emph{orthogonal array} with $s$ \emph{levels},
\emph{strength} $t$ and \emph{index}~$\lambda$, denoted by
$\text{OA}_{\lambda}(N,k,s,t)$, if every $N\times t$
subarray of~$\cA$ contains each $t$-tuple exactly $\lambda$  times as a row.
It is known~\cite{HSS99} that a $\text{TD}_{\lambda}(k,m)$ is
equivalent to an orthogonal array
${\text{OA}_{\lambda}(\lambda\cdot m^2,k,m,2)}$.

A $[k\times (n-k),(n-k)(k-\delta+1),\delta]$ MRD code
$\cC$ is a maximum distance separable (MDS)
code if it is viewed  as a code of length $k$  over
$\F_{q^{n-k}}$~\cite{Gab85}. Thus its codewords form an orthogonal array
$\text{OA}_{\lambda}(q^{(n-k)(k-\delta+1)},k,q^{n-k},k-\delta+1)$
with $\lambda=1$, which is also an orthogonal array
$\text{OA}_{\lambda}(q^{(n-k)(k-\delta+1)},k,q^{n-k},2)$ with
$\lambda=q^{(n-k)(k-\delta-1)}$ (see~\cite{HSS99} for the
connection between MDS codes and orthogonal arrays).

By the equivalence of transversal designs and orthogonal arrays,
and by Theorem~\ref{trm:CDC=TD},
an $(n,k,\delta)_q$ code $\CMRD$
induces an
$\text{OA}_{\lambda}(q^{(n-k)(k-\delta+1)},\frac{q^k-1}{q-1},q^{n-k},2)$
with ${\lambda=q^{(n-k)(k-\delta-1)}}$.
These parameters are different from the ones obtained by viewing
an MRD code as an MDS code.

Now we define a new type of transversal designs in terms of
subspaces, which will be called a subspace transversal design. We
will show that such a design is induced by the codewords of a
lifted MRD code. Moreover, we will show that this design is useful
to obtain upper bounds on the codes that contain the lifted MRD
codes, and in a construction of large constant dimension codes.

Let $\mathbb{V}_0^{(n,k)}$ be a set of one-dimensional subspaces in
$\mathcal{G}_q(n,1)$, that contains only  vectors  starting with
$k$ \emph{zeroes}. Note that $\mathbb{V}_0^{(n,k)}$ is isomorphic to
$\mathcal{G}_q(n-k,1)$.

A \emph{subspace transversal  design} of groupsize $q^m$, $m=n-k$,
block dimension $k$, and \emph{strength} $t$, denoted by
$\text{STD}_{q}(t, k, m)$, is a triple
$(\mathbb{V}^n,\mathbb{G},\mathbb{B})$, where

\begin{enumerate}
\item $\mathbb{V}^n$ is the subset of all elements of
$\mathcal{G}_q(n,1)\setminus \mathbb{V}_0^{(n,k)}$,
$|\mathbb{V}^n|=\frac{(q^{k}-1)}{q-1} q^{m}$ (the \emph{points});

\item $\mathbb{G}$ is a partition of $\mathbb{V}^n$ into
$\frac{q^{k}-1}{q-1}$ classes of size $q^m$ (the \emph{groups});

\item $\mathbb{B}$ is a collection of $k$-dimensional
subspaces which contain only points from $\mathbb{V}^n$ (the
\emph{blocks});

\item each block meets each group in exactly one point;

\item each $t$-dimensional subspace (with points from $\mathbb{V}^n$) which
meets each group in at most one point is contained in exactly
one block.
\end{enumerate}

An $\text{STD}_{q}(t, k, m)$ is \emph{resolvable} if the set $\mathcal{B}$
can be partitioned into sets $\mathcal{B}_1,...,\mathcal{B}_s$,
where each one-dimensional subspace of $V$ is contained in exactly one block of
each $\mathcal{B}_i$. The sets $\mathcal{B}_1,...,\mathcal{B}_s$ are called
\emph{parallel classes}.

As a direct consequence form Lemma~\ref{lm:k-delta+1} and
Theorem~\ref{trm:CDC=TD} we infer the following theorem.

\begin{theorem}\label{trm:MRD=STD}
The codewords of an $(n,k,\delta)_q$
$\CMRD$ form the blocks of a resolvable
$\text{STD}_{q}(k-\delta+1,k, n-k)$, with the set of points
$\mathbb{V}^n$ and the set of groups $\mathbb{V}_A^{(n,k)}$,
$A\in\mathcal{G}_q(k,1)$, defined previously in this section.
\end{theorem}
\begin{remark}
There is no known nontrivial $q$-analog of a block design with
$\lambda =1$ and $t >1$. An $\text{STD}_{q}(t, k, m)$ is very
close to such a design.
\end{remark}
\begin{remark}
An $\text{STD}_{q}(t, k, n-k)$ cannot exist if $k > n-k$, unless
$t=k$. This is not difficult to prove and we leave it as an
exercise for the interested reader. Recall, that the case $k >
n-k$ was not considered in this section (see
Theorem~\ref{trm:param lifted MRD}).
\end{remark}

\section{Upper bounds on the size of codes in $\Gr$}
\label{sec:upbound} In this section we consider upper bounds on
the size of constant dimension codes. First, in
Subsection~\ref{subsec:known bounds} we consider the Johnson type
upper bound presented in~\cite{EV,EV08,WXS-N03,XiFu09}. We estimate
the size of known constant dimension codes relatively to this
bound. The estimations provide better results than the ones
known before, e.g.~\cite{KK}.
In Subsection~\ref{subsec:upper bounds contained MRD} we
provide new upper bounds on codes which contain lifted MRD codes.
This type of upper bounds was not considered before, even so,
as said before, usually the largest known codes contain the
lifted MRD codes.

\subsection {Some known upper bounds}
\label{subsec:known bounds}
Upper bounds on the sizes of constant dimension codes were
obtained in several papers, e.g.~\cite{KK,SKK08}.
The following upper bound was established in~\cite{WXS-N03} in the context of
linear authentication codes and in~\cite{EV,EV08,XiFu09} based on anticodes
in the Grassmannian and as generalization of the well known
Johnson bound for constant weight codes.
\begin{theorem}
\label{thm:Johnson}
\begin{equation}
\mathcal{A}_{q}(n,2\delta,k)\leq\frac{\sbinomq{n}{k-\delta+1}}{\sbinomq{k}{k-\delta+1}}.
\label{eq:Johnson}
\end{equation}
\end{theorem}

It was proved recently~\cite{BlEt11} that for fixed $q$, $k$, and $\delta$,
the ratio between the upper bound of Theorem~\ref{thm:Johnson}
and $\mathcal{A}_{q}(n,2\delta,k)$ equals to 1 as $n \rightarrow \infty$.
But, the method used in~\cite{BlEt11} is based on probabilistic arguments
and an explicit construction of the related code is not known. We will
estimate the value of this upper bound.

$$
\frac{\sbinomq{n}{k-\delta+1}}{\sbinomq{k}{k-\delta+1}}=
\frac{(q^{n}-1)(q^{n-1}-1)\ldots(q^{n-k+\delta}-1)}{(q^{k}-1)(q^{k-1}-1)\ldots(q^{\delta}-1)}
$$
$$
=q^{(n-k)(k-\delta+1)}\frac{(1-q^{-n})(1-q^{-n+1})\ldots(1-q^{-n+k-\delta})}
{(1-q^{-k})(1-q^{-k+1})\ldots(1-q^{-\delta})}
$$
$$ < \frac{q^{(n-k)(k-\delta+1)}}{\prod_{j=\delta}^{\infty}(1-q^{-j})}.
$$
We define $Q_\delta(q)=\prod_{j=\delta}^{\infty}(1-q^{-j})$, $\delta \geq 1$.
Similar analysis for $Q_1(q)$ was considered in~\cite{KK}
and $Q_2(q)$ was considered also in~\cite{GaYa10a}.
Since $(n,k,\delta)_q$ $\C^{\textmd{MRD}}$ has $q^{(n-k)(k-\delta+1)}$ codewords
we have that

\begin{lemma}
\label{lm:ratio} The ratio between the size of an $(n,k,\delta)_q$
$\C^{\textmd{MRD}}$
and the upper bound on $\cA_q(n,2\delta,k)$ given
in~(\ref{eq:Johnson}) satisfies
$$\frac{|\C^\textmd{{MRD}}|}{\Gauss{n}{k-\delta+1}/\Gauss{k}{k-\delta+1}}
> Q_{\delta}(q).
$$
\end{lemma}
\vspace{0.1cm}

The function $Q_\delta(q)$ is increasing in $q$ and also in $\delta$. In
Table~\ref{tab:Qs}, we provide several values of $Q_\delta(q)$ for
different $q$ and $\delta$. For $q=2$ these values were given
in~\cite{Ber80}.

\begin{table}[h]
\centering \caption{ $Q_{\delta}(q)$} \label{tab:Qs}
\begin{tabular}{|c|c|c|c|c|c|}
  \hline
  \backslashbox{$\delta$}{q} & 2 & 3 & 4 & 5&7  \\ \hline
  $2$ & 0.5776 & 0.8402 & 0.9181& 0.9504 & 0.9763 \\
  \hline
  $3$ & 0.7701 & 0.9452 & 0.9793 & 0.9900 & 0.9966\\
  \hline
  $4$ &  0.8801 & 0.9816 & 0.9948 & 0.9980 & 0.9995\\
   \hline
  $5$ &  0.9388 & 0.9938 & 0.9987 & 0.9996& 0.9999 \\
  \hline
\end{tabular}
\end{table}

One can verify that for $q$ large enough or for $\delta$ large enough
the size of a lifted MRD code approaches the upper
bound~(\ref{eq:Johnson}). Thus, an improvement on the lower bound
of $\cA_q(n,2\delta,k)$ is mainly important for small minimum distance
and small $q$. This will be the line of research in the following
sections.

Note, that the lower bound of Lemma~\ref{lm:ratio} is not precise
for small values of $k$. But, it is better improved by another
construction, the \emph{multilevel construction}~\cite{EtSi09}.
For example, for $\delta=2$, the lower bound on the ratio between the size of a
constant dimension code $\C^{\textmd{ML}}$ generated by the
multilevel construction and the upper bound
on $\cA_q(n,2\delta,k)$ given in~(\ref{eq:Johnson}), is presented
in Table~\ref{tab:ratio C^ML}. The values in the table
are larger than the related values in Table~\ref{tab:Qs}.
In the construction of such a code
$\C^{\textmd{ML}}$ we consider only $\CMRD$ code and the codewords
related to the following three \emph{identifying vectors}
(see~\cite{EtSi09} or Section~\ref{sec:construct} for the
definitions)
$\underset{k-2}{\underbrace{11...1}}0011\underset{n-k-2}{\underbrace{000...00}}$,
$\underset{k-3}{\underbrace{11...1}}010101\underset{n-k-3}{\underbrace{000...00}}$,
and
$\underset{k-2}{\underbrace{11...1}}000011\underset{n-k-4}{\underbrace{000...00}}$,
which constitute most of the code. But, since not all identifying
vectors were taken in the computations the values in Table~\ref{tab:ratio C^ML}
are only lower bounds on the ratio, rather than the exact ratio.

\begin{table}[h]
\centering \caption{Lower bounds on ratio between
$|\C^{\textmd{ML}}|$ and the bound in~(\ref{eq:Johnson})}
\label{tab:ratio C^ML}
\begin{tabular}{|c|c|c|c|c|c|}
  \hline
  \backslashbox{k}{q} & $2$ & $3$ & $4$ & $5$ & $7$  \\ \hline
  $3$ & 0.7101 & 0.8678 & 0.9267 & 0.9539 & 0.9771 \\\hline
  $4$ & 0.6657 & 0.8571 & 0.9231 & 0.9524 & 0.9767 \\\hline
  $8$ & 0.6274 & 0.8519 & 0.9219 & 0.9520 & 0.9767 \\\hline
  $30$& 0.6250 & 0.8518 & 0.9219 & 0.9520 & 0.9767 \\
  \hline
\end{tabular}
\end{table}
\vspace{-0.2cm}

\subsection{Upper bounds for codes which contain lifted MRD codes}
\label{subsec:upper bounds contained MRD}
In this subsection we
will derive upper bounds on the size of a constant dimension code which
contains the lifted MRD code $\C^{\text{MRD}}$.

Let $\mathbb{T}$ be a subspace transversal design derived from $(n,k,\delta)_q$
$\C^{\text{MRD}}$ by Theorem \ref{trm:MRD=STD}. Recall that $\cL^{(n,k)}$
is the set of $q^{n}-q^{n-k}$ vectors of length $n$ over $\F_q$ in
which not all the first $k$ entries are {\it zeroes}. Let $\cL_0^{(n,k)}$
be the set of vectors in $\F_q^n$ which start with $k$
\emph{zeroes}. $\cL_0^{(n,k)}$ is isomorphic to $\F_q^{n-k}$,
$|\cL_0^{(n,k)}|=q^{n-k}$,
and $\F_q^n=\cL_0^{(n,k)}\cup\cL^{(n,k)}$. Note, that
$\mathbb{V}_0^{(n,k)}$ is the set of one-dimensional subspaces of
$\mathcal{G}_q(n,1)$ which contain only vectors from $\cL_0^{(n,k)}$. A
codeword of a constant dimension code, in $\mathcal{G}_q(n,k)$,
contains one-dimensional subspaces from $\mathcal{G}_q(n,1)=
\mathbb{V}_0^{(n,k)}\cup \mathbb{V}^n$. Let $\mathbb{C}$ be a constant
dimension code such that $\C^{\text{MRD}}\subset \mathbb{C}$. Each
codeword of $\C\setminus\C^{\text{MRD}}$ contains either at least
two points from the same group of $\mathbb{T}$ or only points from
$\mathbb{V}_0^{(n,k)}$ and hence it contains vectors of $\cL_0^{(n,k)}$.

\begin{theorem}\label{trm:upper bound from Steiner Structure}
If an $(n,M,2(k-1),k)_q$ code $\mathbb{C}$, $k \geq 3$, contains
an $(n,k,k-1)_q$ $\CMRD$ then $M\leq
q^{2(n-k)}+ \cA_q (n-k,2(k-2),k-1)$.
\end{theorem}
\begin{proof}
Let $\mathbb{T}$ be an $\text{STD}_q(2,k, n-k)$  obtained from an
$(n,k,k-1)_q$ $\C^{\text{MRD}} \subset \C$.
Since the minimum distance of $\mathbb{C}$ is $2(k-1)$, it follows
that any two codewords of $\mathbb{C}$ intersect in at most one
one-dimensional subspace. Hence, each two-dimensional subspace of
$\F_q^n$ is contained in at most one codeword of $\mathbb{C}$.
Each two-dimensional subspace $X$ of $\F_q^n$, such that $X=\Span
{ \{v,u \} }$, $v\in \mathbb{V}_A^{(n,k)}$, $u\in \mathbb{V}_B^{(n,k)}$, where
$A\neq B$, $A,B\in \mathcal{G}_q(k,1)$, is contained in a codeword
of $\C^{\text{MRD}}$ by Corollary~\ref{cor:two_elements}
(or by Theorem~\ref{trm:MRD=STD}). Hence, each
codeword $X\in \mathbb{C\setminus}\C^{\text{MRD}}$ either contains
only points from $\mathbb{V}_0^{(n,k)}$ or contains points from
$\mathbb{V}_0^{(n,k)}$ and points from $\mathbb{V}_A^{(n,k)}$, for some $A\in
\mathcal{G}_q(k,1)$. Clearly, $\textmd{dim}(X\cap\cL_0^{(n,k)})=k$ in the
first case and $\textmd{dim}(X\cap\cL_0^{(n,k)})=k-1$ in the second case.
Since $k \geq 3$ and two codewords of $\mathbb{C}$ intersect in at
most a one-dimensional subspace, it follows that each
$(k-1)$-dimensional subspace of $\cL_0^{(n,k)}$ can be contained only in
one codeword. Moreover, since the minimum distance of the code is
$2(k-1)$, it follows that if $X_1 , X_2 \in
\mathbb{C\setminus}\C^{\text{MRD}}$ and $\textmd{dim}(X_1
\cap\cL_0^{(n,k)})= \textmd{dim}(X_2 \cap\cL_0^{(n,k)})=k-1$ then $d_S (X_1
\cap\cL_0^{(n,k)}, X_2 \cap\cL_0^{(n,k)}) \geq 2(k-2)$. Therefore, $\C' \deff \{
X\cap\cL_0^{(n,k)} ~:~ X \in \mathbb{C\setminus}\C^{\text{MRD}},~
\textmd{dim}(X\cap\cL_0^{(n,k)})=k-1 \}$ is an $(n-k,M',2(k-2),k-1)_q$
code. Let $\dS$ be the set of codewords in
$\mathbb{C\setminus}\C^{\text{MRD}}$ such that $\textmd{dim}(X
\cap\cL_0)=k$. For each $X \in \dS$ let $\tilde{X}$ be an
arbitrary $(k-1)$-dimensional subspace of $X$, and let $\dS' \deff
\{ \tilde{X} ~:~ X \in \dS \}$ (note that $|\dS'|=|\dS|$).
Since $d_S(\C') \geq 2(k-2)$, $k \geq 3$, and
each two codewords of~$\mathbb{C}$ intersect in at most
a one-dimensional subspace, it follows that the
code $\C' \cup \dS'$ is an $(n-k,M'',2(k-2),k-1)_q$ code.
This implies the result of the theorem.
\end{proof}

\begin{theorem}
\label{trm:bound 2k-k} If an $(n,M,2k,2k)_q$ code $\mathbb{C}$
contains an $(n,2k,k)_q$ $\CMRD$ then
$M\leq q^{(n-2k)(k+1)}+ \Gauss{n-2k}{k} \frac{q^n-q^{n-2k}}{q^{2k}-q^k} +
\cA_q (n-2k, 2k,2k)$.
\end{theorem}

\begin{proof}
Let $\mathbb{T}$ be an $\text{STD}_q(k+1,2k,n-2k)$  obtained from
an $(n,2k,k)_q$ $\CMRD \subset
\C$. Since the minimum distance of $\mathbb{C}$ is~$2k$, it
follows that any two codewords of $\mathbb{C}$ intersect in at
most a $k$-dimensional subspace. Hence, each $(k+1)$-dimensional
subspace of $\F_q^n$ is contained in at most one codeword of
$\mathbb{C}$. Each $(k+1)$-dimensional subspace $Y$ of~$\F_q^n$,
such that $Y=\Span { \{v_1,...,v_k ,v_{k+1}\} }$, $v_i\in
\mathbb{V}_{A_i}^{(n,2k)}$, where $A_i\neq A_j$, for $i \neq j$, and $A_i
\in \mathcal{G}_q(2k,1)$, $1 \leq i \leq k+1$, is contained in a
codeword of $\C^{\text{MRD}}$ by Theorem~\ref{trm:MRD=STD}. Hence,
each codeword $X\in \mathbb{C}\setminus\mathbb{C}^{MRD}$ has a
nonempty intersection with exactly $\frac{q^{k-\tau}-1}{q-1}$
groups of $\T$, for some $0\leq \tau \leq k$ and therefore $\dim
(X \cap \cL_0^{(n,2k)})=k+\tau$. Let $\dS_\tau$ be the set of codewords
defined by, $X \in \dS_\tau$ if $\dim (X \cap \cL_0^{(n,2k)}) =k+\tau$.

The set $\dS_k$ forms an $(n-2k,M',2k,2k)_q$ code and hence
$|\dS_k| \leq \cA_q (n-2k, 2k,2k)$.

Let $Y$ be a $k$-dimensional subspace of $\cL_0^{(n,2k)}$. If $X_1$ and
$X_2$ are two codewords which contain $Y$ then $Y = X_1 \cap X_2$.
Let $N_{\tau,Y}$ be the number of codewords from $\dS_\tau$ which
contain~$Y$. Clearly, for each $\tau$, $0 \leq \tau \leq k$, we
have
\begin{equation}
\label{eq:NtauY} \sum_{Y \in \mathcal{G}_q(n-2k,k)} N_{\tau,Y} =
|\dS_\tau| \Gauss{k+\tau}{k}~. \end{equation}

There are $\frac{q^n-q^{n-2k}}{q-1}$ points in $\mathbb{V}^n$ and
each $X \in \dS_\tau$ contains exactly
$\frac{q^{2k}-q^{k+\tau}}{q-1}$ points from $\mathbb{V}^n$. Hence,
each ${k\text{-dimensional}}$ subspace $Y$ of $\cL_0^{(n,2k)}$ can be a
subspace of at most $\frac{q^n-q^{n-2k} - \sum_{\tau=1}^{k-1}
N_{\tau,Y} (q^{2k}-q^{k+\tau})  }{q^{2k}-q^k}$ codewords of
$\dS_0$.

Therefore,

\begin{equation*}
| \C | \leq q^{(n-2k)(k+1)}+\sum_{\tau=1}^{k}|\dS_\tau|
\end{equation*}
\begin{equation*}
+ \sum_{Y \in \mathcal{G}_q(n-2k,k)} \frac{q^n-q^{n-2k} -
\sum_{\tau=1}^{k-1} N_{\tau,Y} (q^{2k}-q^{k+\tau}) }{q^{2k}-q^k}
\end{equation*}
\begin{equation*}
= q^{(n-2k)(k+1)}+\sum_{\tau=1}^{k}|\dS_\tau|+ (\Gauss{n-2k}{k}
\frac{q^n-q^{n-2k}}{q^{2k}-q^k}
\end{equation*}
\begin{equation*}
-\sum_{\tau=1}^{k-1}|\dS_\tau|\Gauss{k+\tau}{k}\frac{q^{2k}-q^{k+\tau}}{q^{2k}-q^k}),
\end{equation*}
where the equality is derived from~(\ref{eq:NtauY}).

\vspace{0.1cm}

One  can  easily  verify  that
$\Gauss{k+\tau}{k}\frac{q^{2k}-q^{k+\tau}}{q^{2k}-q^k}\geq1$ for
${1\leq\tau\leq k-1}$; recall also that $|\dS_k| \leq \cA_q (n-2k,
2k,2k)$; thus we have
\begin{small}
\begin{equation*}
|\C|\leq q^{(n-2k)(k+1)}+\Gauss{n-2k}{k}
\frac{q^n-q^{n-2k}}{q^{2k}-q^k}+\cA_q (n-2k, 2k,2k).
\end{equation*}
\end{small}
\end{proof}

\vspace{-0.2cm}
\section{Constructions for $(n,M,4,3)_q$ codes}
\label{sec:construct}

In this section we  discuss and present a construction of codes
which contain an $(n,k,\delta)$ $\C^{\text{MRD}}$ and attain the bound of
Theorem~\ref{trm:upper bound from Steiner Structure}. Such a
construction is  presented only for ${k=3}$ and $q$ large enough. If
$q$ is not large enough then codes obtained by a modification
of this construction almost attain the bound. In any case the codes
obtained in this section are the largest ones known for $k=3$ and $\delta =2$.

For $k=3$, the upper bound of Theorem~\ref{trm:upper bound from
Steiner Structure} on the size of a code which contains
an $(n,3,2)_q$ $\C^{\text{MRD}}$ is ${q^{2(n-3)}+\Gauss{n-3}{2}}$. The construction
which follows is inspired by the construction methods described
in~\cite{EtSi09} and~\cite{TrRo10}. The construction is based on
representation of subspaces by
Ferrers diagrams, optimal rank-metric codes, pending dots, and
one-factorization of the complete graph. The definitions and results
of the first  subsection are taken from~\cite{EtSi09},~\cite{vLWi92},
and~\cite{TrRo10}.

\vspace{-0.2cm}

\subsection{Preliminaries for the construction}

\subsubsection {Representation of subspaces}

For each $X\in\Gr$ represented by the generator matrix in reduced
row echelon form, denoted by $\mbox{RE}(X)$, we associate a binary vector
of length $n$ and weight $k$, $v(X)$, called the \emph{identifying
vector} of $X$, where the \emph{ones} in $v(X)$ are exactly in the
positions where $\mbox{RE}(X)$ has the leading coefficients (the
pivots).
All the  binary vectors of length $n$ and weight $k$
can be considered as  the identifying vectors of all the subspaces
in $\Gr$. These $\binom{n}{k}$ vectors partition  $\Gr$
into the $\binom{n}{k}$ different classes, where each class
consists of all subspaces in $\Gr$ with the same identifying
vector.

The {\it Ferrers tableaux form} of a subspace $X$, denoted by
$\cF(X)$, is obtained from $\mbox{RE}(X)$ first by removing from
each row of $\mbox{RE}(X)$ the {\it zeroes} to the left of the
leading coefficient; and after that removing the columns which
contain the leading coefficients. All the remaining entries are
shifted to the right. The \emph{Ferrers diagram} of $X$, denoted
by $\cF_X$, is obtained from $\cF(X)$ by replacing the entries of
$\cF(X)$ with dots. Given $\cF(X)$, the unique corresponding
subspace $X\in \Gr$ can be easily found.

\begin{example}  Let $X$ be the subspace in $\mathcal G_2(7,3)$ with
the  following  generator matrix in reduced row echelon form:
$$\mbox{RE}(X)=\left( \begin{array}{ccccccc}
\textbf{1} & \color{red}0 & 0 & 0 & \color{red} 1  & \color{red} 1& \color{red} 0\\
0 & 0 & \textbf{1} & 0 & \color{red} 1 & \color{red}0 & \color{red}1 \\
0 & 0 & 0 &  \textbf{1} &\color{red} 0& \color{red} 1 & \color{red} 1
\end{array}
\right) ~.$$
Its identifying vector is $v(X)=1011000$, and its Ferrers tableaux
form and Ferrers diagram are given by

$$\begin{array}{cccc}
0 & 1 & 1 & 0 \\
&1 & 0 & 1  \\
&0 & 1 & 1
\end{array}~~~
\; \textrm{and }\;
\begin{array}{cccc}
 \bullet & \bullet & \bullet & \bullet \\
  & \bullet & \bullet & \bullet   \\
  & \bullet & \bullet & \bullet  \\
\end{array},\; \textrm{respectively }.$$
\end{example}

\subsubsection {Lifted Ferrers diagram rank-metric codes}

Let $\cF$ be a Ferrers diagram with $k$ dots in the rightmost
column and $\ell$ dots in the top row. A code $\cC_{\cF}$ is an
$[\cF,\varrho,\delta]$ {\it Ferrers diagram rank-metric code} if
all codewords of $\cC_{\cF}$ are $k\times \ell$ matrices in which
all entries not in $\cF$ are {\it zeroes}, it forms a rank-metric
code with dimension $\varrho$ and minimum rank distance $\delta$.
The following result is the direct consequence from Theorem 1
in~\cite{EtSi09}.

\begin{lemma}\label{lm:FD MRD bound}
Let $n\geq 8$, $k=3$, $\delta=2$, and let $v$ be an identifying
vector, of length $n$ and weight three,
in which the leftmost \emph{one} appears in one of the first three
entries. Let $\cF$ be the corresponding Ferrers diagram and
$[\cF,\varrho,2]$ be a Ferrers diagram rank-metric code. Then
$\varrho$ is at most the number of dots in $\cF$, which are not
contained in its first row.
\end{lemma}

A code which attains the bound of Lemma~\ref{lm:FD MRD bound} will
be called a Ferrers diagram MRD code. A construction for such
codes can be found in~\cite{EtSi09}.

For a codeword $A \in \cC_{\cF} \subset \F_q^{k \times (n-k)}$,
let $A_{\cF}$ denotes the part of
$A$ related to the entries of $\cF$ in $A$. Given a Ferrers
diagram MRD code $\cC_{\cF}$, a lifted Ferrers diagram MRD code
$\C_{\cF}$ is defined as follows:

$$\C_{\cF} = \{X\in \Gr :
\cF(X)=A_{\cF},~ A \in \cC_{\cF} \}.
$$

This definition is the generalization of the definition of a lifted
MRD code. The following lemma~\cite{EtSi09} is the generalization of the result
given in Theorem~\ref{trm:param lifted MRD}.

\begin{lemma}
\label{lem:dist_lift} If $\cC_{\cF} \subset \F_q^{k \times (n-k)}$
is an $[ \cF , \varrho , \delta ]$ Ferrers diagram
rank-metric code, then its lifted code $\C_{\cF}$ is an $(n, q^\varrho ,
2\delta , k)_q$ constant dimension code.
\end{lemma}

\subsubsection {The multilevel construction and pending dots}

It was proved in~\cite{EtSi09} that for any two subspaces
$X,Y\in\Gr$ we have $d_S(X,Y)\geq d_H(v(X),v(Y))$, where $d_H$
denotes the Hamming distance; and  if $v(X)=v(Y)$ then
$d_S(X,Y)=2d_R(\mbox{RE}(X),\mbox{RE}(Y))$. These properties of
the subspace distance were used in~\cite{EtSi09} to present a
multilevel construction, for a constant dimension code $\C$. In
this construction, first a binary constant weight code $C$ of
length $n$, weight $k$, and minimum Hamming distance $2 \delta$ is
chosen. The codewords of $C$ will serve as the identifying vectors
for~$\C$. For each identifying vector a corresponding lifted
Ferrers diagram MRD code with minimum rank distance $\delta$
is constructed. The union of these lifted Ferrers diagram MRD
codes is an $(n, M, 2\delta, k)_q$ code.

In the construction which follows, for $\delta =2$, we also use a multilevel method,
i.e., we first choose a binary constant weight code $C$ of length
$n$, weight $k=3$, and minimum Hamming distance $2\delta -2=2$.
For each codeword in $C$ a corresponding lifted Ferrers diagram
MRD code is constructed. However, since for some pairs of
identifying vectors the Hamming distance is 2, we need to use
appropriate lifted Ferrers diagram MRD codes to make sure that the
final subspace distance of the code will be 4. For this purpose we
use a method based on pending dots in a Ferrers
diagram~\cite{TrRo10}.

The \emph{pending dots} of a Ferrers diagram $\cF$ are the leftmost dots in the
first row of $\cF$ whose removal has no impact on the size of the
corresponding Ferrers diagram rank-metric code. The following lemma follows from~\cite{TrRo10}.

\begin{lemma}\cite{TrRo10}\label{lm:pending dots}
Let $X$ and $Y$  be two subspaces in $\Gr$ with
$d_H(v(X),v(Y))=2\delta-2$, such that the leftmost \emph{one} of
$v(X)$ is in the same position as the leftmost  \emph{one} of
$v(Y)$.  Let $P_X$ and $P_Y$ be the sets of pending dots of $X$ and $Y$, respectively.
If $P_X\cap P_Y\neq \varnothing$
and the entries in $P_X\cap P_Y$ (of their Ferrers tableaux forms) are
assigned with different values in at least one position, then
$d_S(X,Y)\geq 2\delta.$
\end{lemma}

\begin{example}
Let $X$ and $Y$  be subspaces in $\mathcal{G}_q(8,3)$ which are given by
the following generator matrices:

$$\mbox{RE}(X)=\left(\begin{array}{cccccccc}
1 &\textcircled{\raisebox{-0.9pt}{0}} &\textcircled{\raisebox{-0.9pt}{0}} & 0  & v_1 & v_2 & 0 & v_3\\
0 & \:0 &\: 0 & 1  & v_4 & v_5 & 0 & v_6 \\
0 & \:0 & \:0 & 0 & 0 & 0 &  1 & v_7
\end{array}\right)
$$
$$\mbox{RE}(Y)=\left(\begin{array}{cccccccc}
1 &\textcircled{\raisebox{-0.9pt}{0}} &\textcircled{\raisebox{-0.9pt}{1}} & v'_1 & 0 & v'_2 & 0 & v'_3\\
0 & \:0 & \:0 & 0 & 1  & v'_4 & 0 & v'_5 \\
0 & \:0 & \:0 & 0 & 0 & 0 &  1 & v'_6
\end{array}\right),
$$
where $v_i ,v'_i\in \F_q$, and the pending dots are emphasized by
circles. Their identifying vectors are $v(X)=10010010$ and
$v(Y)=10001010$. Clearly, $d_H(v(X), v(Y))=2$, while $d_S(X,Y)=4$.
\end{example}

\subsubsection{One-factorization of complete graphs}

A \emph{matching} in a graph $G$ is a set of pairwise disjoint
edges of $G$.  A \emph{one-factor}  is a matching  such that
every vertex of $G$ occurs in exactly one edge of the matching. A
partition of the edge set in $G$ into one-factors is called a
\emph{one-factorization}. Let $K_n$ be a complete graph with $n$
vertices. The following lemma is a well known result~\cite[p. 476]{vLWi92}.

\begin{lemma}\label{lm:one-factor even}
$K_{2n}$ has a one-factorization for all $n$.
\end{lemma}

A \emph{near-one-factor} in $K_{2n-1}$ is a matching with $n-1$ edges
which contain all but one vertex. A set of near-one-factors which
contains each edge in $K_{2n-1}$ precisely once is called a
\emph{near-one-factorization}. The following corollary is the direct
consequence from Lemma \ref{lm:one-factor even}.

\begin{corollary}
$K_{2n-1}$ has a near-one-factorization for all $n$.
\end{corollary}

\begin{corollary}\label{lm:1-factorization}
Let $D$ be a set of all binary vectors  of length~$m$ and weight
$2$.
\begin{itemize}
\item If $m$ is even, $D$ can be partitioned into $m-1$ classes,
each one has $\frac{m}{2}$ vectors with pairwise disjoint positions of \emph{ones};
\item If $m$ is odd, $D$ can be partitioned into $m$ classes,
each one has $\frac{m-1}{2}$ vectors with pairwise disjoint positions of \emph{ones}.
\end{itemize}
\end{corollary}

\subsection{The first construction}

\textbf{Construction I:} Let $n\geq 8$ and $q^2+q+1\geq n-4$ for
odd  $n$ (or $q^2+q+1\geq n-3$ for even $n$).

\subsubsection{Identifying vectors}
The identifying
vector $v_0=11100\ldots 0$ corresponds to the lifted MRD code
$\C^{\text{MRD}}$. The other identifying vectors are of the form
$x||y$, where $x$ is of length 3 and weight one, and $y$ is of length
$n-3$ and weight two. We use all the $\binom{n-3}{2}$ vectors of
weight two in the last $n-3$ coordinates of the identifying vectors.
By Corollary~\ref{lm:1-factorization}, there is a partition of the
set of vectors of length $n-3$  and weight 2 into $s=n-4$ classes
if $n-3$ is even (or into $s=n-3$ classes if $n-3$ is odd),
$P_1,P_2,\ldots,P_s$.
 We define
\[\cA_1=\{(001)||y:y\in P_1\},\]
\[\cA_2=\{(010)||y:y\in P_i, 2\leq i\leq \min\{q+1,s\}\},\]
\[\cA_3= \left\{\begin{array}{cc}
                \{(100)||y:y\in P_i,~ q+2\leq i\leq s\} & \textmd{if }s>q+1 \\
                \varnothing & \textmd{if }s\leq q+1 \\
              \end{array}\right..
\]

\subsubsection{Ferrers tableaux forms and pending dots}

All the Ferrers diagrams  which correspond to the identifying
vectors from~$\cA_2$ have  one common pending dot in the first
entry of the first row. We assign the same value of $\F_q$ in this
entry of the Ferrers tableaux form for each vector in the same
class. Two subspaces with identifying vectors from different
classes of~$\cA_2$ have different values in the entry of this
pending dot. This is possible since the number of classes in~$\cA_2$
is at most $q$. On the remaining dots of Ferrers diagrams we
construct Ferrers diagram MRD codes and lift them.

Similarly, all the Ferrers diagrams which correspond to the
identifying vectors from $\cA_3$, have two common pending dots in
the first two entries of the first row. We assign the same value
of $\F_q$ in these two entries in the Ferrers tableaux form for
each vector in the same class. Two subspaces with identifying
vectors from different classes of $\cA_3$ have different values in
at least one of these two entries. This is possible
since the number of classes in $\cA_3$
is at most $q^2$. On the remaining dots of
Ferrers diagrams we construct Ferrers diagram MRD codes and lift
them.

Finally, we lift Ferrers diagrams MRD codes which correspond
to the identifying vectors of $\cA_1$.

\subsubsection{The code}

Our code $\C$ is a union of $\C^{\text{MRD}}$ and the lifted codes
corresponding to the identifying vectors in $\cA_1$, $\cA_2$, and~$\cA_3$.
\begin{example}
\label{ex:construction_k=3}
For $n=8$, there are $\binom{5}{2}$ different binary vectors of length $8-3=5$ and weight $2$.
We partition these vectors into five disjoint classes $P_1=\left\{(11000),(00110)\right\}$,
$P_2=\left\{(10100),(01001)\right\}$, $P_3=\left\{(10010),(00101)\right\}$,
$P_4=\left\{(10001),(01010)\right\}$, $P_5=\left\{(01100),(00011)\right\}$.
The identifying vectors of the code, besides $v_0=11100000$, are partitioned into three sets,
$$\mathcal{A}_1=\left\{(00111000), (00100110) \right\},$$
$$\mathcal{A}_2=\left\{(01010100), (01001001),(01010010) ,(01000101) \right\},$$
$$\mathcal{A}_3=\left\{(10010001), (10001010), (10001100), (10000011) \right\}.$$
To demonstrate the idea of the construction we will only
consider the set $\mathcal{A}_2$. The generator matrices in reduced
row echelon form of the codewords with identifying vectors from~$\mathcal{A}_2$
are of four different types:
$$\left( \begin{array}{cccccccc}
0& 1 & v^1_1 & 0 & v^1_2 & 0  & v^2_3 & v^1_4\\
0 & 0 & 0& 1 & v^1_5 & 0 & v^1_6 & v^1_7\\
0 & 0 & 0 & 0 & 0 & 1 & v^1_8 &v^1_9
\end{array}
\right),
$$
$$
 \left( \begin{array}{cccccccc}
0 & 1 & v^2_1  & v^2_2 & 0  & v^2_3 & v^2_4& 0\\
0 & 0 & 0&0& 1 & v^2_5 & v^2_6 &0\\
0 & 0 & 0 & 0 & 0 &0 & 0 & 1
\end{array}
\right),
$$
$$\left( \begin{array}{cccccccc}
0& 1 & v^3_1 & 0 & v^3_2 &  v^3_3 & 0 & v^3_4\\
0 & 0 & 0 & 1 & v^3_5 & v^3_6 & 0  & v^3_7\\
0 & 0 & 0 & 0 & 0   & 0   & 1  & v^3_8
\end{array}
\right),
$$
$$
 \left( \begin{array}{cccccccc}
0 & 1 & v^4_1  & v^4_2  & v^4_3 &  0 &  v^4_4 & 0\\
0 & 0 & 0 & 0 & 0&  1 & v^4_5 &0\\
0 & 0 & 0 & 0 & 0 &0 & 0 & 1
\end{array}
\right),
$$
where all the $v^j_i$'s are elements from $\F_q$.
The suffixes (last $n-3$ coordinates) of the identifying vectors of the first two generator matrices
belong to $P_2$, and of the last two matrices to $P_3$.
All these matrices have the same pending dot
in the place of $v^i_1$, $1 \leq i \leq 4$. Then we
assign $0$ in this place for the two first matrices and $1$ in this place for the two last matrices:
$$\left( \begin{array}{cccccccc}
0& 1 & \textbf{0} & 0 & v^1_2 & 0  & v^1_3 & v^1_4\\
0 & 0 & 0& 1 & v^1_5 & 0 & v^1_6 & v^1_7\\
0 & 0 & 0 & 0 & 0 & 1 & v^1_8 &v^1_9
\end{array}
\right),
$$
$$
 \left( \begin{array}{cccccccc}
0 & 1 & \textbf{0}  & v^2_2 & 0  & v^2_3 & v^2_4& 0\\
0 & 0 & 0&0& 1 & v^2_5 & v^2_6 &0\\
0 & 0 & 0 & 0 & 0 &0 & 0 & 1
\end{array}
\right),
$$
$$\left( \begin{array}{cccccccc}
0& 1 & \textbf{1} & 0 & v^3_2 &  v^3_3 & 0 & v^3_4\\
0 & 0 & 0 & 1 & v^3_5 & v^3_6 & 0  & v^3_7\\
0 & 0 & 0 & 0 & 0   & 0   & 1  & v^3_8
\end{array}
\right),
$$
$$
 \left( \begin{array}{cccccccc}
0 & 1 & \textbf{1}  & v^4_2  & v^4_3 &  0 &  v^4_4 & 0\\
0 & 0 & 0 & 0 & 0&  1 & v^4_5 &0\\
0 & 0 & 0 & 0 & 0 &0 & 0 & 1
\end{array}
\right).
$$

\end{example}

\subsubsection{Analysis of the construction}

\begin{theorem}\label{trm:existence of codes for ContrA} For $q$ satisfying
$q^2+q+1\geq s$, where
\[s= \left\{\begin{array}{cc}
                n-4, & n \textrm{ is odd }\\
                n-3, & n \textrm{ is even }\\
              \end{array},\right.
\]
the code $\C$ obtained in Construction I attains the bound of
Theorem~\ref{trm:upper bound from Steiner Structure}.
\end{theorem}
\begin{proof}
First, we prove that the minimum subspace distance of $\C$ is 4.
Let $X, Y\in \C$, $X \neq Y$. We distinguish between three cases:
\begin{itemize}
\item Case 1: If $X,Y\in \C^{\text{MRD}}$ then
$d_S(X,Y)\geq 4$ since the minimum distance of the $(n,3,2)_q$ $\CMRD$ is 4.
\item Case 2: If $X\in \C^{\text{MRD}}$ and $Y\in
\mathbb{C\setminus}\C^{\text{MRD}}$ then
$d_S(X,Y) \geq d_H(v(X),v(Y)) \geq 4$.
\item Case 3: Assume $X,Y\in \mathbb{C\setminus}\C^{\text{MRD}}$.

If $v(X)\in \cA_i$, $v(Y)\in
\cA_j$, $i\neq j$, then clearly $d_S(X,Y) \geq d_H(v(X),v(Y)) \geq 4$.

If $v(X),v(Y)\in \cA_i$, i.e.,
$X$ and $Y$ have identifying vectors $v(X) = z||w$, $v(Y)=z||w'$, where $z$ is of length~3,
we distinguish between two additional cases:
\begin{itemize}
\item $w,w' \in P_i$, $1 \leq i \leq s$.
In this case $d_H(v(X),v(Y))=4$ which implies $d_S(X,Y)\geq 4$.
\item $w \in F_i$, $w' \in F_j$, $i \neq j$.
If $d_H(v(X),v(Y))=4$ then $d_S(X,Y)\geq 4$. If $d_H(v(X),v(Y))=2$
then by Lemma~\ref{lm:pending dots} we have that
$d_S(X,Y){\geq4}$.
\end{itemize}
\end{itemize}

Next, we calculate the size of $\C$. Recall, that the identifying
vectors are partitioned into $s$ classes. Note that since ${q^2+q+1\geq s}$,
it follows that each one of the $\binom{n-3}{2}$ vectors of weight~$2$
and length $n-3$ is taken as the suffix of some identifying
vector.
Each such suffix (of length $n-3$ and weight $2$)
is the identifying vector of a subspace in $\mathcal{G}_q(n-3,2)$.
By Lemma~\ref{lm:FD MRD bound} each such subspace in $\mathcal{G}_q(n-3,2)$
is contained in exactly one codeword (since the first row of
the generator matrix of
the 3-dimensional subspace is omitted
by the lemma for the bound on $\varrho$). The size of $\CMRD$
is $q^{2(n-3)}$ and the size of $\mathcal{G}_q(n-3,2)$
is $\Gauss{n-3}{2}$. Hence, the size
of $\C$ is $q^{2(n-3)}+\Gauss{n-3}{2}$.
Theorem~\ref{trm:upper bound from Steiner Structure} implies that
for $(n,M,4,3)_q$ code $\mathbb{C}$, which contains
an $(n,3,2)_q$ $\CMRD$ we have $M\leq
q^{2(n-3)}+ \cA_q (n-3,2,2) = q^{2(n-3)}+\Gauss{n-3}{2}$.
\end{proof}

\begin{remark}
A $(6,M,4,3)_q$ code whose size attains the upper bound of
Theorem~\ref{trm:upper bound from Steiner Structure} was
constructed in~\cite{EtSi09} and a $(7,M,4,3)_q$ code whose size
attains this bound was constructed in~\cite{TrRo10}.
\end{remark}

\subsection{The second construction}

For small alphabets
Construction~I is modified as follows.

\textbf{Construction II:} Let $n\geq 8$ and $q^2+q+1 < n-4$ for
odd $n$ (or $q^2+q+1 < n-3$ for even $n$).

The identifying vector $v_0=11100\ldots 0$ corresponds to the
lifted MRD code $\C^{\text{MRD}}$. Let
$\alpha=\floorenv{\frac{n-3}{q^2+q+2}}$ and
$r=n-3-\alpha(q^2+q+2)$. For each other identifying vector,
we partition the last $n-3$ coordinates
into $\alpha$ or $\alpha +1$ sets, where each one of
the first $\alpha$ sets consists of $q^2+q+2$
consecutive coordinates and the last set (which exists if
$r >0$) consists of $r<q^2+q+2$
consecutive coordinates. Since $q^2+q+2$ is always an even
integer, it follows from Corollary~\ref{lm:1-factorization} that
there is a partition of vectors of length $q^2+q+2$ and weight 2,
corresponding to the $i$th set,  $1\leq i\leq \alpha$, into
$q^2+q+1$ classes $P_1^i,P_2^i,\ldots,P_{q^2+q+1}^i$. We define
$Y_1^i=\{0^{(i-1)(q^2+q+2)}||y||0^{n-3-i(q^2+q+2)}:y\in P_1^i\}$,
$Y_2^i=\{0^{(i-1)(q^2+q+2)}||y||0^{n-3-i(q^2+q+2)}:y\in P_j^i, 2\leq
j\leq q+1\}$, and
$Y_3^i=\{0^{(i-1)(q^2+q+2)}||y||0^{n-3-i(q^2+q+2)}:y\in P_j^i,
q+2\leq j\leq q^2+q+1\}$, where $0^{\ell}$ denotes the \emph{zeroes} vector of length $\ell$. Let
$$\cA_1^i=\{(001)||y:y\in Y_1^i\},\;1\leq i\leq \alpha,$$
$$\cA_2^i=\{(010)||y:y\in Y_2^i\},\;1\leq i\leq \alpha,$$
$$\cA_3^i=\{(100)||y:y\in Y_3^i\},\;1\leq i\leq \alpha.$$

The identifying vectors (excluding $v_0$), of the code that we
construct, are partitioned into the following three sets:
$$\cA_1=\cup_{i=1}^{\alpha}\cA_1^i,~
\cA_2=\cup_{i=1}^{\alpha}\cA_2^i,~
\cA_3=\cup_{i=1}^{\alpha}\cA_3^i.$$

As in Construction I, we construct a lifted Ferrers diagram MRD
code for each identifying vector, by using pending dots. Our code
$\C$ is a union of $\C^{\text{MRD}}$ and the lifted codes
corresponding to the identifying vectors in $\cA_1$, $\cA_2$, and
$\cA_3$.

\vspace{0.2cm}

\begin{remark} The identifying vectors
with two \emph{ones} in the last $r$ entries can be also used in
Construction II, but their contribution to the final code is
minor.
\end{remark}

In a similar way to the proof of Theorem~\ref{trm:existence of
codes for ContrA} one can prove the following theorem, based on
the fact that the size of the lifted Ferrers diagram MRD code
obtained from the identifying vectors in $\cA_1^i\cup \cA_2^i\cup
\cA_3^i$, $1\leq i\leq \alpha$, is
$\Gauss{q^2+q+2}{2}q^{2(n-3-(q^2+q+2)i)}$.

\begin{theorem}\label{trm:codes for ContrA'} For $q$ satisfying
$q^2+q+1<s$, where

\[s = \left\{\begin{array}{cc}
                n-4, & n \textrm{ is odd }\\
                n-3, & n \textrm{ is even } \\
              \end{array}\right.,
\]
Construction II generates an $(n,M,4,3)_q$ constant dimension code
with
$M=q^{2(n-3)}+\sum_{i=1}^{\alpha}\Gauss{q^2+q+2}{2}q^{2(n-3-(q^2+q+2)i)}$,
which  contains an $(n,3,2)_q$ $\C^\textmd{{MRD}}$.
\end{theorem}

For all admissible values of $n$, the ratio
$(|\C| - |\CMRD|) /$$\begin{tiny}\left[\begin{array}{c} n-3\\
2\end{array}\right]_{q}\end{tiny}$, for the code $\C$ generated by
Construction~II, is greater than 0.988 for $q=2$ and 0.999 for
$q>2$. Hence, the code almost attains the bound of
Theorem~\ref{trm:upper bound from Steiner Structure}.

In the following table we compare the size of codes obtained by
Constructions I and II (denoted by $\C_{new}$) with the size of the
largest previously known codes (denoted by $\C_{old}$) and with the
upper bound~(\ref{eq:Johnson}) (for $k=3$).

\vspace{0.15cm}
\begin{table}[h]
\centering
\begin{tabular}{|c|c|c|c|c|}
\hline $q$ & $n$ &  $|\C_{old}|$ & $|\C_{new}|$ &upper bound~(\ref{eq:Johnson})
\tabularnewline \hline \hline 2 & 13 & $1192587$~\cite{EtSi09} &
$1221296$&1597245
 \tabularnewline \hline
2 & 14 &  $4770411$~\cite{EtSi09} & $4885184$& 6390150
\tabularnewline \hline 5 & 9 &  $244644376$~\cite{EtSi09}&
$244649056$& 256363276 \tabularnewline \hline
\end{tabular}
\end{table}

\vspace{0.15cm}

The new ratio between the new best lower bound and the upper
bound~(\ref{eq:Johnson}) with $k=3$
and $\delta=2$, is presented in Table~\ref{tab:full new ratio}.
One should compare it with Table~\ref{tab:ratio C^ML}.

\begin{table}[h]
\centering \caption{The ratio between $|\C_{new}|$ and
the bound in~(\ref{eq:Johnson})} \label{tab:full new ratio}
\begin{tabular}{|c|c|c|c|c|c|}
  \hline
  \backslashbox{k}{q} & 2 & 3 & 4 & 5 & 7  \\ \hline
  $3 $ & 0.7657 & 0.8738 & 0.928 & 0.9543 & 0.9772  \\
  \hline
\end{tabular}
\end{table}

\vspace{-0.2cm}


\section {Construction for $(8,M,4,4)_q$ codes}
\label{sec:constructionB}

In this section we introduce a construction of $(8,M,4,4)_q$ codes
which attain the upper bound of Theorem~\ref{trm:bound 2k-k}, and are
the largest codes with these parameters.
This construction is based on 2-parallelism of subspaces in
$\mathcal{G}_q(4,2)$.

A $k$\emph{-spread} in $\Gr$ is a set of $k$-dimensional subspaces
which partition $\F_q^n$ (excluding the all-zero vector). We say
that two subspaces are disjoint if they have only trivial
intersection. A $k$-spread in $\Gr$  exists if and only if $k$
divides $n$~\cite{ScEt02}. Clearly, a $k$-spread is a constant dimension code in
$\Gr$ with the maximal possible minimum distance $d_S=2k$. A partition
of all $k$-dimensional subspaces of $\Gr$ into disjoint
$k$-spreads is called a $k$\emph{-parallelism}. The following
construction is presented for $q=2$.

\textbf{Construction III:} Let $\mathbb{T}$ be an
$\text{STD}_2(3,4,4)$ obtained from an $(8,4,2)_2$
$\C^{\text{MRD}}$. We will generate a new code $\C$ which
contains $\CMRD$.
The following new codewords (blocks) will form the elements of
$\mathbb{C\setminus}\C^{\text{MRD}}$.

Let $\cB_1, \cB_2,\ldots, \cB_7$ be a partition of all the
subspaces of $\mathcal{G}_2(4,2)$ into seven $2$-spreads, each one
of size 5, i.e., a well known 2-parallelism in
$\mathcal{G}_2(4,2)$~\cite{Bak76,Beu74,ZZS71}. For each $i$, $1\leq i\leq 7$,
and each two subspaces $Z,Z'\in \cB_i$
($Z'$ can be equal to $Z$) we write $Z=\{v_0=\textbf{0},
v_1, v_2, v_3\}$ and $Z'=\{ v'_0=\textbf{0}, v'_1, v'_2, v'_3\}$,
where $ v_t, v'_t \in \F_2^4$, ${0\leq t\leq 3}$, and
$\textbf{0}=(0000)$. The $2$-dimensional subspace $Z$ has four
cosets  $Z_0=Z, Z_1, Z_2, Z_3$ in $\F_2^4$.
We construct the following four codewords in
$\mathbb{C\setminus}\C^{\text{MRD}}$. The codewords are defined by fifteen
nonzero vectors which are the nonzero vectors of a $4$-dimensional subspace
as can be verified.

\begin{enumerate}
\item[(C.1)]
$\{(\textbf{0}||u): u \in Z_0 \setminus \{\textbf{0}\}\}\cup
\{(v'_1||y):y\in Z_0 \}$
\begin{align*}\hspace*{-2cm}
\cup
\{(v'_2||y):y\in Z_0 \}\cup
\{(v'_3||y):y \in Z_0 \},
\end{align*}
\item[(C.2)]
$\{(\textbf{0}||u): u \in Z_0 \setminus \{\textbf{0}\}\}\cup
\{(v'_1||y):y\in Z_1\}$
\begin{align*}\hspace*{-2cm}
\cup
\{(v'_2||y):y\in Z_2\}\cup
\{(v'_3||y):y\in Z_3\},
\end{align*}
\item[(C.3)]
$\{(\textbf{0}||u): u \in Z_0\setminus \{\textbf{0}\}\}\cup
\{(v'_1||y):y\in Z_2\}$
\begin{align*}\hspace*{-2cm}
\cup
\{(v'_2||y):y\in Z_3\}\cup
\{(v'_3||y):y\in Z_1\},
\end{align*}
\item[(C.4)]
$\{(\textbf{0}||u): u \in Z_0 \setminus \{\textbf{0}\}\}\cup
\{(v'_1||y):y\in Z_3\}$
\begin{align*}\hspace*{-2cm}
\cup
\{(v'_2||y):y\in Z_1\}\cup
\{(v'_3||y):y\in Z_2\}.
\end{align*}

%
%
\end{enumerate}
In $\mathcal{G}_2(4,2)$ there are $\GaussBin{4}{2} =35$~~$2$-dimensional
subspaces, and hence there are 35 different choices for $Z$. Since
the size of a spread is 5, it follows that there are 5 different
choices for $Z'$. Thus, there are a total of $35 \cdot 5 \cdot 4 =700$
codewords in $\C \setminus \CMRD$ generated in this way.
In addition to these 700 codewords we add a codeword which contains
all the points of $\mathbb{V}_0^{(n,k)}$.

\begin{example}
A partition of $\mathcal{G}_2(4,2)$ into seven spreads is given in the following table, where each row corresponds to a spread.

\begin{footnotesize}
\begin{table}[h]
\centering \caption{Partition of $\mathcal{G}_2(4,2)$}
\label{tab:partition}
\begin{tabular}{|c|c|c|c|c|c|}
  \hline
  $\cB_1$ & $\begin{array}{c}
            1000 \\
            0100 \\
            1100
          \end{array}$  & $\begin{array}{c}
            1010 \\
            0101 \\
            1111
          \end{array}$ & $\begin{array}{c}
            1011 \\
            0110 \\
            1101
          \end{array}$ & $\begin{array}{c}
            1001 \\
            0111 \\
            1110
          \end{array}$
          & $\begin{array}{c}
            0010 \\
            0001 \\
            0011
          \end{array}$\\
    \hline
  $ \cB_2$ & $\begin{array}{c}
            1000 \\
            0010 \\
            1010
          \end{array}$ & $\begin{array}{c}
            0100 \\
            0001 \\
            0101
          \end{array}$ & $\begin{array}{c}
            1011 \\
            0111 \\
            1100
          \end{array}$ & $\begin{array}{c}
            1001 \\
            0110 \\
            1111
          \end{array}$ & $\begin{array}{c}
            1101 \\
            0011 \\
            1110
          \end{array}$ \\
   \hline
  $ \cB_3$ & $\begin{array}{c}
            1000 \\
            0110 \\
            1110
          \end{array}$ & $\begin{array}{c}
            1001 \\
            0100 \\
            1101
          \end{array}$ & $\begin{array}{c}
            1100 \\
            0011 \\
            1111
          \end{array}$ & $\begin{array}{c}
            0101 \\
            0010 \\
            0111
          \end{array}$ & $\begin{array}{c}
            1010 \\
            0001 \\
            1011
          \end{array}$ \\
  \hline
  $ \cB_4$ & $\begin{array}{c}
            1000 \\
            0001 \\
            1001
          \end{array}$ & $\begin{array}{c}
            1011 \\
            0100 \\
            1111
          \end{array}$ & $\begin{array}{c}
            1100 \\
            0010 \\
            1110
          \end{array}$ & $\begin{array}{c}
            1010 \\
            0111 \\
            1101
          \end{array}$ & $\begin{array}{c}
            0101 \\
            0011 \\
            0110
          \end{array}$ \\
   \hline
  $ \cB_5$ & $\begin{array}{c}
            1000 \\
            0101 \\
            1101
          \end{array}$ & $\begin{array}{c}
            0100 \\
            0011 \\
            0111
          \end{array}$ & $\begin{array}{c}
            1010 \\
            0110 \\
            1100
          \end{array}$ & $\begin{array}{c}
            1001 \\
            0010 \\
            1011
          \end{array}$ & $\begin{array}{c}
            1110 \\
            0001 \\
            1111
          \end{array}$ \\
    \hline
  $ \cB_6$ & $\begin{array}{c}
            1000 \\
            0111 \\
            1111
          \end{array}$ & $\begin{array}{c}
            0100 \\
            0010 \\
            0110
          \end{array}$ & $\begin{array}{c}
            1100 \\
            0001 \\
            1101
          \end{array}$ & $\begin{array}{c}
            1000 \\
            0011 \\
            1001
          \end{array}$ & $\begin{array}{c}
            1011 \\
            0101 \\
            1110
          \end{array}$ \\
   \hline
  $ \cB_7$ & $\begin{array}{c}
            1000 \\
            0011 \\
            1011
          \end{array}$ & $\begin{array}{c}
            1010 \\
            0100 \\
            1110
          \end{array}$ & $\begin{array}{c}
            1001 \\
            0101 \\
            1100
          \end{array}$ & $\begin{array}{c}
            1101 \\
            0010 \\
            1111
          \end{array}$ & $\begin{array}{c}
            0110 \\
            0001 \\
            0111
          \end{array}$ \\
  \hline
\end{tabular}
\end{table}
\end{footnotesize}
\vspace{0.3cm}

We illustrate the idea of Construction III by considering one 2-spread and a coset of one element of the spread.
Let $\cB_1 = \{ Z^0 ,Z^1, Z^2,Z^3,Z^4 \}$ be a spread given by the first row of the table, i.e.,
$Z^0=\Span{(1000),(0100)}$, $Z^1=\Span{(1010),(0101)}$, $Z^2=\Span{(1011),(0110)}$,
$Z^3=\Span{(1001),(0111)}$, $Z^4=\Span{(0010),(0001)}$.
 The four cosets of $Z^0$ are given by
$$\color{OXO-emph}Z_0 = Z^0=\{(0000),(1000),(0100),(1100)\},$$
$$\color{red}Z_1=\{(0001),(1001),(0101),(1101)\},$$
$$\color{blue}Z_2=\{(0010),(1010),(0110),(1110)\},$$
$$\color{dark_green}Z_3=\{(0011),(1011),(0111),(1111)\}.$$

For the pair $Z^0, Z^1$, the following four subspaces $C_1$, $C_2$,
$C_3$, and $C_4$, belong to the code and correspond to the four types of the codewords,
where $C_i$ corresponds to $(\textmd{C}.i)$, ${1\leq i\leq 4}$, and for every coset of $Z^0$ we use a different color.
\vspace{0.1cm}
\begin{small}
\begin{table}[h]
\centering
\begin{tabular}{|c|c|c|c|}
  \hline
  $C_1$ & $C_2$ & $C_3$ & $C_4$ \\
  \hline
  0000\color{OXO-emph}1000 & 0000\color{OXO-emph}1000 & 0000\color{OXO-emph}1000 & 0000\color{OXO-emph}1000\\
  0000\color{OXO-emph}0100 & 0000\color{OXO-emph}0100 & 0000\color{OXO-emph}0100 & 0000\color{OXO-emph}0100\\
  0000\color{OXO-emph}1100 & 0000\color{OXO-emph}1100 & 0000\color{OXO-emph}1100 & 0000\color{OXO-emph}1100\\
 1010\color{OXO-emph}0000 & 1010\color{OXO-emph}0001 & 1010\color{OXO-emph}0010 & 1010\color{OXO-emph}0011\\
 1010\color{OXO-emph}1000 & 1010\color{red}1001 & 1010\color{blue}1010 & 1010\color{dark_green}1011\\
 1010\color{OXO-emph}0100 & 1010\color{red}0101 & 1010\color{blue}0110 & 1010\color{dark_green}0111\\
 1010\color{OXO-emph}1100 & 1010\color{red}1101 & 1010\color{blue}1110 & 1010\color{dark_green}1111\\
 0101\color{OXO-emph}0000 & 0101\color{blue}0010 & 0101\color{dark_green}0011 & 0101\color{red}0001\\
 0101\color{OXO-emph}1000 & 0101\color{blue}1010 & 0101\color{dark_green}1011 & 0101\color{red}1001\\
0101\color{OXO-emph}0100 & 0101\color{blue}0110 & 0101\color{dark_green}0111 & 0101\color{red}0101\\
0101\color{OXO-emph}1100 & 0101\color{blue}1110 & 0101\color{dark_green}1111 & 0101\color{red}1101\\
 1111\color{OXO-emph}0000 & 1111\color{dark_green}0011 & 1111\color{red}0001 & 1111\color{blue}0010\\
1111\color{OXO-emph}1000 & 1111\color{dark_green}1011 & 1111\color{red}1001 & 1111\color{blue}1010\\
 1111\color{OXO-emph}0100 & 1111\color{dark_green}0111 & 1111\color{red}0101 & 1111\color{blue}0110\\
 1111\color{OXO-emph}0100 & 1111\color{dark_green}1111 & 1111\color{red}1101 & 1111\color{blue}1110\\
  \hline
\end{tabular}
\end{table}
\end{small}

\end{example}

\vspace{0.1cm}

\begin{theorem}
\label{trm:construction b} Construction III generates an
$(8,2^{12}+701,4,4)_2$ constant dimension code $\C$ which attains
the bound of Theorem~\ref{trm:bound 2k-k} and contains an
$(8,4,2)_2$ $\CMRD$.
\end{theorem}
\begin{proof}
First, we observe that the four types of codewords given in the
construction are indeed $4$-dimensional subspaces of $\F_2^8$.
Each one of the codewords contains 15 different one-dimensional
subspaces, and hence each codeword contains 15 different nonzero
vectors of $\F_2^8$. It is easy to verify that all these vectors
are closed under addition in $\F_2$, thus each constructed
codeword is a $4$-dimensional subspace of $\F_2^8$.

To prove that for each two codewords $X,Y \in \C$, we have
$d_S(X,Y)\geq 4$, we distinguish between three cases:
\begin{itemize}
\item Case 1: $X,Y\in \C^{\text{MRD}}$. Since the minimum distance
of $\C^{\text{MRD}}$ is 4, we have that $d_S(X,Y)\geq 4$.
\item Case 2: $X\in \C^{\text{MRD}}$ and $Y\in \mathbb{C\setminus}\C^{\text{MRD}}$.
The codewords of $\C^{\text{MRD}}$ forms the
blocks of an $\text{STD}_2(3,4,4)$, $\T$, and hence meet each group in
exactly one point. Each codeword of $\mathbb{C\setminus}\C^{\text{MRD}}$
meets exactly three groups of $\T$. Hence, $\dim(X\cap Y)\leq 2$ for
each $X\in \C^{\text{MRD}}$ and $Y\in \mathbb{C\setminus}\C^{\text{MRD}}$,
therefore,  $d_S(X,Y)\geq 4$.
\item Case 3: $X,Y\in \mathbb{C\setminus}\C^{\text{MRD}}$.
If $X$ and $Y$ have exactly three points in common
in $\mathbb{V}_0^{(8,4)}$ (which correspond to a 2-dimensional subspace
contained in $\cL_0^{(8,4)}$), then they are disjoint in all the groups of
$\mathbb{T}$. This is due to the fact that the points of $X$ in $\mathbb{V}^8$
and the point of $Y$ in
$\mathbb{V}^8$ correspond to either
different cosets, or different blocks in the same spread.
If $X$ and $Y$ have exactly one point in common in $\mathbb{V}_0^{(8,4)}$, then they have at most
two points in common in at most one group of $\T$. Thus, $d_S(X,Y)\geq 4$.
\end{itemize}

$(8,4,2)_2$ $\C^{\text{MRD}}$ contains $2^{12}$ codewords. As
explained in the construction, there are 701 codewords in $\C \setminus \CMRD$.
Thus, in the constructed code $\mathbb{C}$ there are
$2^{12}+701=4797$ codewords.

Thus, the code attains the bound of Theorem~\ref{trm:bound 2k-k}.
\end{proof}

\begin{remark} Theorem~\ref{trm:construction b} implies that $A_2(8,4,4)\geq 4797$
(the previous best known lower bound was $A_2(8,4,4)\geq 4605$~\cite{SiEt10}).
\end{remark}

\begin{remark}\label{rm:gen q} Construction III can be easily generalized for all prime powers $q\geq2$,
since there is a $2$-parallelism  in $\mathcal{G}_q(n,2)$ for all
such $q$, where  $n$ is power of $2$~\cite{Beu74}. Thus, from this
construction we can obtain a $(8,M,4,4)_q$ code with
$M=q^{12}+\Gauss{4}{2}(q^2+1)q^2+1$, since the size of a
$2$-spread in $\mathcal{G}_q(4,2)$ is $q^2+1$ and there are $q^2$
different cosets of a $2$-dimensional subspace in $\F_q^4$.
\end{remark}

In the following table we compare the size of codes obtained by
Construction III and its generalizations for large $q$ (denoted by $\C_{new}$) with the size of the
largest previously known codes (denoted by $\C_{old}$) and with the
upper bound~(\ref{eq:Johnson}) (for $n=8$ and $k=4$).

\vspace{0.15cm}
\begin{table}[h]
\centering
\begin{tabular}{|c|c|c|c|}
\hline $q$ &   $|\C_{old}|$ & $|\C_{new}|$&upper bound~(\ref{eq:Johnson})
\tabularnewline \hline \hline
2 &  $2^{12}+509$~\cite{SiEt10} & $2^{12}+701$&$2^{12}+2381$
\tabularnewline \hline
3 &  $3^{12}+8137$~\cite{EtSi09} & $3^{12}+11701$&$3^{12}+95941$
\tabularnewline \hline
4 &  $4^{12}+72529$~\cite{EtSi09} & $4^{12}+97105$&$4^{12}+1467985$
 \tabularnewline \hline

\end{tabular}
\end{table}
\vspace{0.15cm}

\begin{remark}
In general, the existence of $k$-parallelism  in $\Gr$ is an open
problem. It is known that $2$-parallelism exists for $q=2$ and
all $n$~\cite{Bak76,ZZS71}, and for each prime power~$q$,
where $n$ is power of $2$~\cite{Beu74}. There is also a
$3$-parallelism for $q=2$ and $n=6$~\cite{Sar02}. Thus we believe
that Construction III can be generalized to a larger family of
parameters assuming that there exists a corresponding parallelism.
\end{remark}


\section{Linear codes derived from lifted MRD codes}
\label{sec:Linear codes}
A lifted MRD code and the transversal design derived from it
can also be used to construct a linear code in the Hamming space.
In this section we study the properties of such a linear code,
whose parity-check
matrix is an incidence matrix of a transversal design derived from a lifted MRD code.
Some of the results presented in this section generalize
the results given in~\cite{JoWe04}. In particular, the lower
bounds on the minimum distance and the bounds on the dimension of
codes derived from lifted MRD codes with $k-\delta=1$ coincide
with the bounds on  LDPC codes from partial geometries considered
in~\cite{JoWe04}.

For each codeword $X$ of an $(n,k,\delta)_q$
$\C^{\textmd{MRD}}$ we define its binary \emph{incidence vector}
$x$ of length $|\mathbb{V}^n|=\frac{q^{n}-q^{n-k}}{q-1}$ as follows:
$x_A=1$ if and only if the point (one-dimensional
subspace) $A \in \mathbb{V}^n$ is contained
in $X$.

Let $H$ be the $|\C^{\textmd{MRD}}|\times|\mathbb{V}^n|$ binary
matrix whose rows are the incidence vectors of the codewords of
$\C^{\textmd{MRD}}$. By Theorem~\ref{trm:CDC=TD}, this matrix $H$
is the \emph{incidence matrix} of a
$\text{TD}_{\lambda}(\frac{q^k-1}{q-1},\;q^{n-k})$, with
$\lambda=q^{(n-k)(k-\delta-1)}$. Note that the rows of the
incidence matrix $H$ correspond to the blocks of the transversal
design, and the columns of $H$ correspond to the points of the
transversal design. If $\lambda=1$ in such a design (or,
equivalently, $\delta=k-1$ for $\C^{\textmd{MRD}}$), then $H^T$ is
an incidence matrix of a \emph{net}, the dual structure to the
transversal  design~\cite[p. 243]{vLWi92}.

An $[N,K,d]$ linear code is a linear subspace of dimension $K$
of $\F_2^N$ with minimum Hamming distance $d$.
Let $C$ be the linear code with the parity-check matrix $H$, and let
$C^T$ be the linear code with the parity-check matrix $H^T$.



The code $C$ has length $\frac{q^{n}-q^{n-k}}{q-1}$ and the code $C^T$ has length
$q^{(n-k)(k-\delta +1)}$.
By Corollary \ref{cor:2}, each column of $H$ has $q^{(n-k)(k-\delta)}$
ones; since each $k$-dimensional subspace contains $\frac{q^{k}-1}{q-1}$
one-dimensional subspaces, each row has $\frac{q^{k}-1}{q-1}$ ones.

\begin{remark}
Note that if  $\delta=k$, then the column weight of $H$ is one.
Hence, the minimum distance of $C$ is 2. Moreover, $C^T$ consists
only of the all-zero codeword. Thus, these codes are not
interesting  and hence in the sequel we assume that $\delta\leq
k-1$.
\end{remark}

\begin{lemma}\label{lm:partition to permutation blocks}
The matrix $H$ obtained from an $(n,k,\delta)_q$
$\C^{\textmd{MRD}}$ code can be decomposed into blocks, where each
block is a $q^{n-k}\times q^{n-k}$ permutation matrix.
\end{lemma}

\begin{proof} It follows from Lemma~\ref{lm:resolv}
that the related transversal design is resolvable.
In each parallel class each element of $\mathbb{V}^n$ is contained
in exactly one codeword of $\C^{\textmd{MRD}}$. Each class has
$q^{n-k}$ codewords, each group has $q^{n-k}$ points, and each
codeword meets each group in exactly one point. This implies that
the $q^{n-k}$ rows of $H$ related to each such class can be
decomposed into $\frac{q^k-1}{q-1}$ $\;q^{n-k}\times q^{n-k}$
permutation matrices.
\end{proof}

\begin{example}
A $[12,4,6]$ code $C$ and a $[16,8,4]$ code $C^T$ are obtained
from the $(4,16,2,2)_2$ lifted MRD code $\C^\textmd{{MRD}}$. The
incidence matrix for corresponding transversal design
$\text{TD}_1(3,4)$ (see Example~\ref{ex:TD(3,4)}) is given by the
following $16\times 12$ matrix. The four rows above this matrix
represent the column vectors for the points of the design.
\vspace{0.2cm}

\begin{footnotesize}
$$\begin{array}{c}
 \left. \begin{tabular}{c|c|c}
    \bf {0 0 0 0} & \bf {1 1 1 1} & \bf {1 1 1 1} \\
    \bf {1 1 1 1} & \bf {0 0 0 0} & \bf {1 1 1 1} \\
    \bf {0 0 1 1} & \bf {0 0 1 1} & \bf {0 0 1 1} \\
    \bf {0 1 0 1} & \bf {0 1 0 1} & \bf {0 1 0 1} \\
    \end{tabular} \right. \\
    \left(\begin{tabular}{c|c|c}\hline\hline
    1 0 0 0 & 1 0 0 0 & 1 0 0 0 \\
    0 1 0 0 & 0 0 0 1 & 0 0 1 0 \\
    0 0 1 0 & 0 1 0 0 & 0 0 0 1 \\
    0 0 0 1 & 0 0 1 0 & 0 1 0 0 \\
     \hline
    1 0 0 0 & 0 1 0 0 & 0 1 0 0 \\
    0 1 0 0 & 0 0 1 0 & 0 0 0 1 \\
    0 0 1 0 & 1 0 0 0 & 0 0 1 0 \\
    0 0 0 1 & 0 0 0 1 & 1 0 0 0 \\
     \hline
    1 0 0 0 & 0 0 1 0 & 0 0 1 0 \\
    0 1 0 0 & 0 1 0 0 & 1 0 0 0 \\
    0 0 1 0 & 0 0 0 1 & 0 1 0 0 \\
    0 0 0 1 & 1 0 0 0 & 0 0 0 1 \\
     \hline
    1 0 0 0 & 0 0 0 1 & 0 0 0 1 \\
    0 1 0 0 & 1 0 0 0 & 0 1 0 0 \\
    0 0 1 0 & 0 0 1 0 & 1 0 0 0 \\
    0 0 0 1 & 0 1 0 0 & 0 0 1 0 \\
    \end{tabular}
\right)
\end{array}$$
\end{footnotesize}
\end{example}
\vspace{0.2cm}

\begin{corollary}\label{cor:even weight}
All the codewords of the code $C$, associated with the parity-check matrix
$H$, and of the code $C^T$, associated with the parity-check matrix
$H^T$, have  even weights.
\end{corollary}


\begin{corollary} \label{cor:upperbound on d}
The minimum Hamming distance $d$ of $C$  and the minimum Hamming distance $d^T$
of $C^T$ are upper bounded by $2q^{n-k}$.
\end{corollary}


To obtain a lower bound on the minimum Hamming distance of these
codes we need the following theorem known as the Tanner
bound~\cite{Tan01}.

\begin{theorem}
\label{Tannerbound} The minimum distance, $d_{\textmd{min}}$, of
a linear code defined by an $m\times n$ parity-check matrix
$\mathcal{H}$ with constant row weight $\rho$ and constant column
weight $\gamma$ satisfy

\begin{enumerate}
  \item[{\bf T1}:] $d_{\textmd{min}}\geq
  \frac{n(2\gamma-\mu_{2})}{\gamma\rho-\mu_{2}}$,
  \item[{\bf T2}:] $d_{\textmd{min}}\geq
  \frac{2n(2\gamma+\rho-2-\mu_{2})}{\rho(\gamma\rho-\mu_{2})}$,
\end{enumerate}

\noindent where $\mu_{2}$ is the second largest eigenvalue
of $\mathcal{H}^{T}\mathcal{H}$.
\end{theorem}
\vspace{0.3cm}

To obtain a lower bound on $d$ and $d^T$ we need to find
the second largest eigenvalue of $H^TH$ and $HH^T$, respectively.
Note that since the set of eigenvalues of $H^TH$ and the
set of eigenvalues of $HH^T$ are the same,
it is sufficient to find only the eigenvalues of $H^TH$.

The following lemma is derived from~\cite[p. 563]{HandCD}.

\begin{lemma}\label{lm:spectrum}
Let $\cH$ be an incidence matrix for $\text{TD}_{\lambda}(k,m)$.
The eigenvalues of $\cH^T \cH$ are $rk$, $r$, and $rk-km\lambda$
with multiplicities $1, k(m-1)$, and $k-1$, respectively, where
$r$ is a number of blocks that are incident with a given point.
\end{lemma}
\vspace{0.3cm}

By Corollary \ref{cor:2}, $r=q^{(n-k)(k-\delta)}$ in
$\text{TD}_{\lambda}(\frac{q^k-1}{q-1},\;q^{n-k})$ with
$\lambda=q^{(n-k)(k-\delta-1)}$. Thus, from Lemma
\ref{lm:spectrum} we obtain the spectrum of $H^TH$.

\begin{corollary}
\label{cor:spectrum HTH} The eigenvalues of $H^TH$ are
$q^{(n-k)(k-\delta)}\frac{q^k-1}{q-1}$, $q^{(n-k)(k-\delta)}$, and
$0$ with multiplicities~$1$, $\frac{q^k-1}{q-1}(q^{n-k}-1)$, and
$\frac{q^k-1}{q-1}-1$, respectively.
\end{corollary}

\vspace{0.3cm}

Now, by Theorem~\ref{Tannerbound} and Corollary~\ref{cor:spectrum HTH}, we have

\begin{corollary}\label{cor:mindist}
\[d\geq \frac{q^{n-k}(q^k-1)}{q^k-q} ,
\]
\[d^T\geq \left\{\begin{array}{cc}
                2^k & \delta=k-1,~q=2,~k=n-k \\
                4q^{(n-k)(\delta-k+1)} & \textmd{otherwise} \\
              \end{array}\right. ~.
\]
\end{corollary}

\begin{proof}
By Corollary~\ref{cor:spectrum HTH}, the second largest
eigenvalue of $H^TH$ is $\mu_2=q^{(n-k)(k-\delta)}$. We apply
Theorem~\ref{Tannerbound}({\bf T1}) to obtain

\begin{equation*}
d\geq
\begin{footnotesize}\frac{q^{n-k}\frac{q^k-1}{q-1}(2q^{(n-k)(k-\delta)}-q^{(n-k)(k-\delta)})}
      {q^{(n-k)(k-\delta)}\frac{q^k-1}{q-1}-q^{(n-k)(k-\delta)}}
      =\frac{q^{n-k}(q^k-1)}{q^k-q}~.\end{footnotesize}
\end{equation*}

By using Theorem \ref{Tannerbound} we also obtain lower bounds on
$d^T$:
\begin{equation}
\label{eq:Tanner d^T 1}
d^T\geq \frac{q^{n-k}(2\frac{q^k-1}{q-1}-q^{(n-k)(k-\delta)})}{\frac{q^k-1}{q-1}-1},
\end{equation}
\begin{equation}
\label{eq:Tanner d^T 2}
d^T\geq 4q^{(n-k)(\delta-k+1)}.
\end{equation}

Note that the expression in~(\ref{eq:Tanner d^T 1}) is negative
for $\delta<k-1$. For $\delta=k-1$ with $k=n-k$ and $q=2$, the
bound in~(\ref{eq:Tanner d^T 1}) is larger than the bound
in~(\ref{eq:Tanner d^T 2}). Thus, we have $d^T\geq 2^k$, if
$\delta=k-1,\;q=2$, and $k=n-k$; and $d^T\geq
4q^{(n-k)(\delta-k+1)}$, otherwise.
\end{proof}


%

We use the following result derived from~\cite[Theorem 1]{KV03} to
improve the lower bound on $d^T$.

\begin{lemma}\label{lm:KashVar}
Let $\cH$ be an incidence matrix of blocks (rows) and points
(columns) such that each block contains exactly $\kappa$~points,
and each pair of distinct blocks intersects in at most
$\gamma$~points. If $d_{\cH^T}$ is a minimum distance of a code with the parity-check matrix $\cH^T$ then
\[d_{\cH^T}\geq\frac{\kappa}{\gamma}+1.\]
\end{lemma}

\begin{corollary} $ d^T \geq \frac{q^k-1}{q^{k-\delta}-1}+1$.
\end{corollary}

\begin{proof}
By Lemma \ref{lm:KashVar}, with $\kappa=\frac{q^k-1}{q-1}$ and
$\gamma = \frac{q^{k-\delta}-1}{q-1}$, since any two codewords in
a lifted MRD code intersect in at most $(k-\delta)$-dimensional
subspace, we have the following lower bound on the minimum distance
of $C^T$
$$
d^T\geq \frac{(q^k-1)/(q-1)}{(q^{k-\delta}-1)/(q-1)}+1
=\frac{q^k-1}{q^{k-\delta}-1}+1.
$$
Obviously, for all $\delta \leq k-1$, this bound is larger or
equal than the bound of Corollary \ref{cor:mindist}, and thus
the result follows.
\end{proof}


Let $\dim(C)$  and $\dim(C^T)$ be the dimensions of $C$ and $C^T$,
respectively. To obtain the lower and upper bounds on $\dim(C)$
and $\dim(C^T)$ we need the following basic results from linear
algebra~\cite{HoJo85}. For a matrix $A$ over a field $\F$, let
$\rank_{\F}(A)$ denotes the rank  of  $A$ over $\F$.
\begin{lemma}
\label{lm:algebra}
Let $A$ be a  $\rho \times \eta$ matrix,  and
let $\R$ be the field of real numbers. Then
\begin{itemize}
  \item $\rank_{\R}(A)=\rank_{\R}(A^T)=\rank_{\R}(A^TA)$.
  \item If $\rho=\eta$ and $A$ is a symmetric matrix with the eigenvalue~$0$ of multiplicity $t$,
then $\rank_{\R}(A)=\eta - t$.
\end{itemize}
\end{lemma}

\begin{theorem}\label{trm:lower bound on code dimension}
\[\dim(C)\geq \frac{q^k-1}{q-1}-1,\]
\[\dim(C^T)\geq q^{(n-k)(k-\delta+1)}-\frac{q^k-1}{q-1}(q^{n-k}-1)-1.\]
\end{theorem}

\begin{proof}
First, we observe that
$\dim(C)=\frac{q^k-1}{q-1}q^{n-k}-\rank_{\F_{2}}(H)$, and
$\dim(C^T)=q^{(n-k)(k-\delta+1)}-\rank_{\F_{2}}(H^T)$. Now we
obtain an upper bound on $\rank_{\F_2}(H)=\rank_{\F_2}(H^T)$.
Clearly, $\rank_{\F_{2}}(H)\leq \rank_{\R}(H)$. By
Corollary~\ref{cor:spectrum HTH}, the multiplicity of an
eigenvalue~$0$ of $H^TH$ is $\frac{q^k-1}{q-1}-1$. Hence by
Lemma~\ref{lm:algebra}, $\rank_{\F_2}(H) \leq \rank_{\R}
(H)=\rank_{\R}( H^TH)=
\frac{q^k-1}{q-1}q^{n-k}-(\frac{q^k-1}{q-1}-1)$. Thus,
$\dim(C)\geq \frac{q^k-1}{q-1}q^{n-k}-(\frac{q^k-1}{q-1}q^{n-k}
-(\frac{q^k-1}{q-1}-1))=\frac{q^k-1}{q-1}-1$, and $\dim(C^T)\geq
q^{(n-k)(k-\delta+1)}-\frac{q^k-1}{q-1}q^{n-k}+\frac{q^k-1}{q-1}-1$.
\end{proof}

Now, we obtain an upper bound on the dimension of the codes $C$
and $C^T$ for odd $q$.

\begin{theorem}
\label{trm:upper bound odd q}
Let $q$ be a power of an odd prime number.
\begin{itemize}
  \item If $\frac{q^k-1}{q-1}$
is odd, then
$\dim(C)\leq \frac{q^k-1}{q-1}-1$ and
$\dim(C^T)\leq q^{(n-k)(k-\delta+1)}-\frac{q^k-1}{q-1}(q^{n-k}-1)-1$.

\item If $\frac{q^k-1}{q-1}$ is even, then
$\dim(C)\leq \frac{q^k-1}{q-1}$, and
$\dim(C^T)\leq q^{(n-k)(k-\delta+1)}-\frac{q^k-1}{q-1}(q^{n-k}-1)$.
 \end{itemize}
 \end{theorem}

\begin{proof}
We compute the lower bound on $\rank_{\F_2}(H)$ to obtain the
upper bound on the dimension of the codes $C$ and $C^T$. First, we
observe that $\rank_{\F_2}(H)\geq \rank_{\F_2}(H^TH)$.
By~\cite{BrEi92}, the rank over $\F_2$ of an integral
diagonalizable square matrix $A$ is lower bounded by the sum of
the multiplicities of the eigenvalues of $A$ that do not vanish
modulo $2$. We consider now $\rank_{\F_2}(H^TH)$. By
Corollary~\ref{cor:spectrum HTH}, the second eigenvalue of $H^TH$
is always odd for odd $q$. If $\frac{q^k-1}{q-1}$ is odd, then the
first eigenvalue of $H^TH$ is also odd. Hence, we sum the
multiplicities  of the first two eigenvalues to obtain
$\rank_{\F_2}(H^TH)\geq 1+\frac{q^k-1}{q-1}(q^{n-k}-1)$. If
$\frac{q^k-1}{q-1}$ is even, then the first eigenvalue is even,
and hence we take only the multiplicity of the second eigenvalue
to obtain $\rank_{\F_2}(H^TH)\geq \frac{q^k-1}{q-1}(q^{n-k}-1)$.
The result follows now from the fact that the dimension of a code
is equal to the difference between its length and
$\rank_{\F_2}(H)$.
\end{proof}

\begin{remark}
For even values of $q$ the method used in the proof for
Theorem~\ref{trm:upper bound odd q} leads to a trivial result,
since in this case all the eigenvalues of $H^TH$ are even and thus
by~\cite{BrEi92} we have $\rank_{\F_2}(H^TH)\geq 0$. But clearly,
by Lemma~\ref{lm:partition to permutation blocks} we have
$\rank_{\F_2}(H)\geq q^{n-k}$. Thus, for even $q$, $\dim(C)\leq
\frac{q^k-1}{q-1}q^{n-k}-q^{n-k}=q^{n-k}( \frac{q^k-1}{q-1}-1)$,
and $\dim(C^T)=q^{(n-k)(k-\delta+1)}-q^{n-k}$.
\end{remark}

Note that for odd $q$ and odd $\frac{q^k-1}{q-1}$ the lower and
the upper bounds on the dimension of $C$ and $C^T$ are the same.
Therefore, we have the following corollary.

\begin{corollary} For odd $q$ and odd $\frac{q^k-1}{q-1}$ the dimensions
$\dim(C)$ and $\dim(C^T)$ of the codes $C$ and $C^T$, respectively, satisfy
$\dim(C)=\frac{q^k-1}{q-1}-1$, and
$\dim(C^T)= q^{(n-k)(k-\delta+1)}-\frac{q^k-1}{q-1}q^{n-k}+\frac{q^k-1}{q-1}-1$ .
\end{corollary}

Finally, $C$ and $C^T$ can be also viewed as LDPC
codes obtained from designs~\cite{AHKXL04,JoWe01,JoWe04,KV03,KLF01,LaMi07,LTLMH08,TXLA-G05,VKK02,VaMi04,ZHLA-G10}.
Some preliminary results in this direction can be found in~\cite{Sil11,SiEt11is}.

\section{Conclusions and future research}
\label{sec:conclude}

Lifted MRD codes are considered. Properties of these codes,
especially when viewed as transversal designs are proved.
Based on this design new upper bounds and constructions for constant dimension
codes which contain lifted MRD codes as subcodes are given. The
incidence matrix of the design (which represents also the
codewords of the lifted MRD code) is considered as a parity-check
matrix of a linear code in the Hamming space. Properties of these
linear codes are proved.
We conclude with a list of open problems for future research.
\begin{enumerate}

\item What are the general upper bounds on a size of an
$(n,M,2\delta,k)_q$ code which contains a lifted MRD code?

\item Are the upper bounds of Theorems~\ref{trm:upper bound from
Steiner Structure} and~\ref{trm:bound 2k-k} and related
bounds for other parameters attained for
all parameters?

\item Can the codes constructed in Constructions I, II, and III be
used, in a recursive method, to obtain new bounds on
$\cA_q(n,d,k)$ for larger $n$?

\item One of the main research problems is to improve the lower
bounds on $\cA_q(n,d,k)$, with codes which do not contain the
lifted MRD codes. Only such codes can close the gap between the
lower and the upper bounds on $\cA_q(n,d,k)$ for small $q$ and
small $d$ (e.g. the seven codes for $k=3$ mentioned in the
Introduction).

\item Which properties have LDPC codes obtained from lifted
MRD codes? The bounds given
in Section~\ref{sec:Linear codes} is only a first step
in this direction. In addition, we would like to
know the performance of these codes with various decoding
algorithms~\cite{DPTRU02,Ric03}.
\end{enumerate}


%

%
%

\end{document}